\documentclass[10pt,sigconf]{acmart}

\usepackage{url}
\usepackage{booktabs}
\usepackage{threeparttable}
\usepackage{amsfonts}
\usepackage{nicefrac}
\usepackage{microtype}
\usepackage{graphicx}
\usepackage{subfigure}
\usepackage{amsthm}
\usepackage{natbib}
\usepackage{epstopdf}
\usepackage{grffile}
\usepackage{flushend}
\usepackage{multirow}
\usepackage{tablefootnote}
\usepackage{balance}

\usepackage{amssymb}
\usepackage{amsmath}
\usepackage{bm}

\newtheoremstyle{mydefn}
{}{}
{\it}  
{0pt}       
{\bfseries} 
{:}
{.5em}
{}          
\theoremstyle{mydefn}

\newtheorem{definition}{Definition}[section]
\newtheorem{theorem}{Theorem}[section]
\newtheorem{corollary}{Corollary}[section]
\newtheorem{fact}{Fact}[section]
\newtheorem{lemma}{Lemma}[section]

\newtheoremstyle{myexample}
{}{}
{}  
{0pt}       
{\bfseries} 
{:}
{.5em}
{}          
\theoremstyle{myexample}

\newtheorem{example}{Example}[section]
\theoremstyle{acmplain}

\renewcommand{\paragraph}[1]{\vspace{0.5em}\noindent\textbf{#1.}}
\renewcommand{\subparagraph}[1]{\vspace{0.5em}\noindent\textit{\underline{#1.}}}

\newcommand{\floor}[1]{\left\lfloor #1 \right\rfloor}

\newcommand{\norm}[1]{\left\Vert #1 \right\Vert}
\newcommand{\bracket}[1]{\{ #1 \}}

\makeatletter
\newif\if@restonecol
\makeatother

\usepackage[linesnumbered,ruled,vlined]{algorithm2e}
\usepackage{algpseudocode}

\SetKwComment{Comment}{$\triangleright$\ }{}

\AtBeginDocument{
  \providecommand\BibTeX{{
    \normalfont B\kern-0.5em{\scshape i\kern-0.25em b}\kern-0.8em\TeX}}}

\copyrightyear{2020}
\acmYear{2020}
\setcopyright{acmcopyright}
\acmConference[SIGMOD'20]{Proceedings of the 2020 ACM SIGMOD International Conference on Management of Data}{June 14--19, 2020}{Portland, OR, USA}
\acmBooktitle{Proceedings of the 2020 ACM SIGMOD International Conference on Management of Data (SIGMOD'20), June 14--19, 2020, Portland, OR, USA}
\acmPrice{15.00}
\acmDOI{10.1145/3318464.3389778}
\acmISBN{978-1-4503-6735-6/20/06}

\begin{document}

\title{Locality-Sensitive Hashing Scheme based on Longest Circular Co-Substring}


\author{Yifan Lei}
\affiliation{
  	\institution{National University of Singapore}
}
\email{leiyifan@u.nus.edu}

\author{Qiang Huang}
\authornote{Corresponding author.}
\affiliation{
  	\institution{National University of Singapore}
}
\email{huangq@comp.nus.edu.sg}

\author{Mohan Kankanhalli}
\affiliation{
  	\institution{National University of Singapore}
}
\email{mohan@comp.nus.edu.sg}

\author{Anthony K. H. Tung}
\affiliation{
  	\institution{National University of Singapore}
}
\email{atung@comp.nus.edu.sg}

\renewcommand{\shortauthors}{Y. Lei et al.}

\begin{abstract}
Locality-Sensitive Hashing (LSH) is one of the most popular methods for $c$-Approximate Nearest Neighbor Search ($c$-ANNS) in high-dimensional spaces. In this paper, we propose a novel LSH scheme based on the Longest Circular Co-Substring (LCCS) search framework (LCCS-LSH) with a theoretical guarantee. We introduce a novel concept of LCCS and a new data structure named Circular Shift Array (CSA) for $k$-LCCS search. The insight of LCCS search framework is that close data objects will have a longer LCCS than the far-apart ones with high probability. LCCS-LSH is \emph{LSH-family-independent}, and it supports $c$-ANNS with different kinds of distance metrics. We also introduce a multi-probe version of LCCS-LSH and conduct extensive experiments over five real-life datasets. The experimental results demonstrate that LCCS-LSH outperforms state-of-the-art LSH schemes.

\end{abstract}

%


\maketitle

\vspace{-0.5em}
\section{Introduction}
\label{sect:intro}
Nearest Neighbor Search (NNS) is a fundamental problem, and it has wide applications in various fields, such as data mining, multimedia databases, machine learning, and artificial intelligence. Given a distance metric, a database $\mathcal{D}$ of $n$ data objects and a query $q$ with feature representation in $d$-dimensional space $\mathbb{R}^d$, the aim of NNS is to find the object $o^* \in \mathcal{D}$ which is closest to $q$, where $o^*$ is called the Nearest Neighbor (NN) of $q$. 
The exact NNS in low-dimensional spaces has been well solved by tree-based methods \cite{guttman1984rtree,bentley1990k,katayama1997sr}. 
For high-dimensional NNS, due to the difficulty of finding exact solutions \cite{weber1998quantitative,hinneburg2000nearest}, the approximate version of NNS, named $c$-Approximate NNS ($c$-ANNS), has been widely studied in recent two decades \cite{kleinberg1997two, indyk1998approximate, fagin2003efficient, jagadish2005idistance, beygelzimer2006cover, jegou2010product, sun2014srs, wang2018randomized, malkov2018efficient, zhou2018generic, fu2019fast}. 

\paragraph{Prior Work}
Locality-Sensitive Hashing (LSH) \cite{indyk1998approximate,har2012approximate} and its variants \cite{broder1998min, gionis1999similarity, charikar2002similarity, datar2004locality, panigrahy2006entropy, andoni2006near, gan2012locality, huang2015query, andoni2015optimal, lei2019sublinear} are one of the most popular methods for high-dimensional $c$-ANNS. An LSH scheme consists of two components: the LSH function family (or simply LSH family) and the search framework. The idea of LSH families is to construct a family of hash functions such that the \emph{positive probability} $p_1$ of the close objects to be hashed into the same bucket with a query $q$ is higher than the \emph{negative probability} $p_2$ of the far-apart ones. Furthermore, the search framework aims to increase the gap between $p_1$ and $p_2$, so that the close objects can be identified efficiently. The most popular search frameworks are the static concatenating search framework \cite{datar2004locality, lv2007multi, tao2009quality, liu2014sk} and the dynamic collision counting framework \cite{gan2012locality, huang2015query, zheng2016lazylsh, huang2017query}. 

\begin{figure*}[t]
  \centering
  \subfigure[E2LSH]{
    \label{fig:example:e2lsh}
    \includegraphics[height=3.15cm]{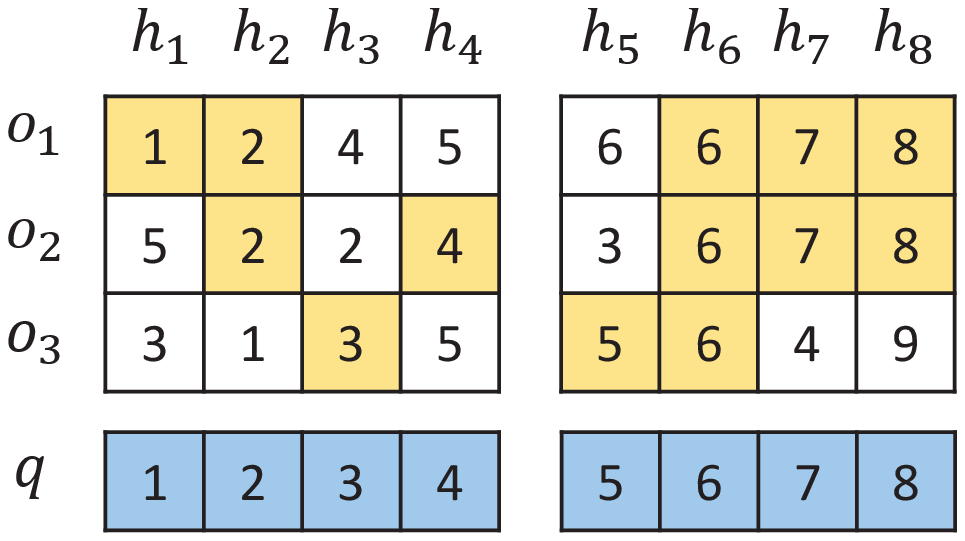}}
  \subfigure[C2LSH]{
    \label{fig:example:c2lsh}
    \includegraphics[height=3.15cm]{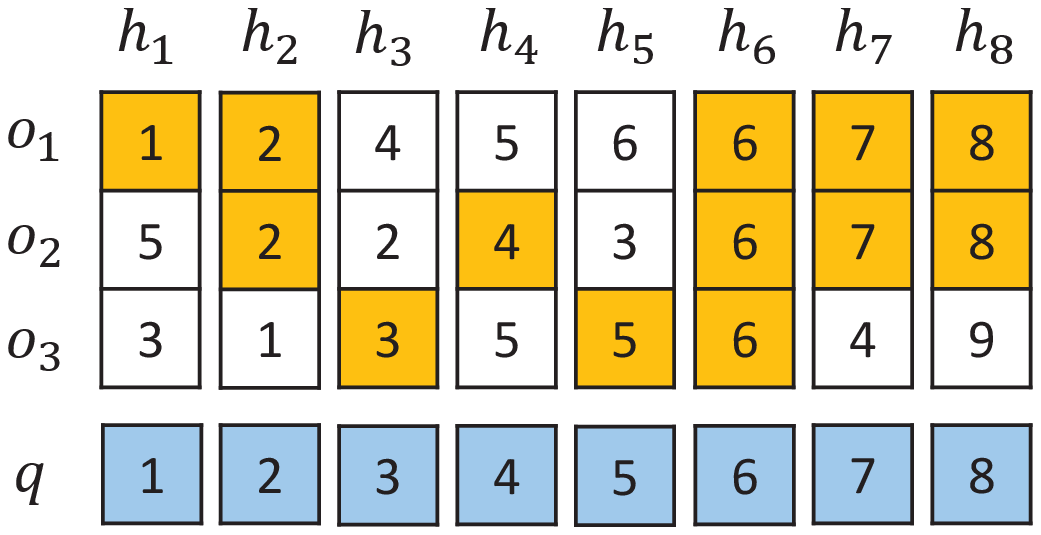}}
  \subfigure[LCCS-LSH]{
    \label{fig:example:lcslsh}
    \includegraphics[height=3.15cm]{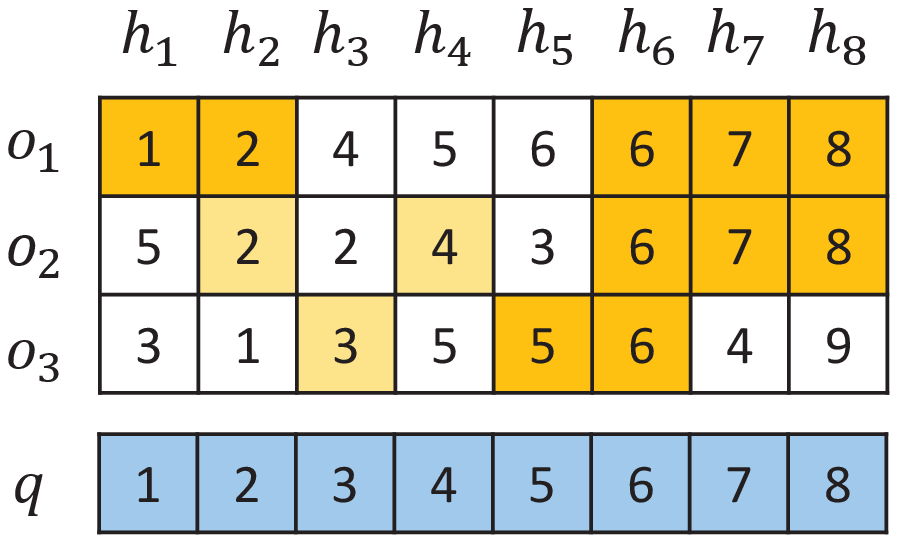}}
  \vspace{-1.0em}
  \caption{An example of the search frameworks of E2LSH, C2LSH, and LCCS-LSH}
  \vspace{-0.75em}
  \label{fig:example}
\end{figure*}

\subparagraph{Static Concatenating Search Framework}
The static concatenating search framework was first introduced by Indyk et al. \cite{indyk1998approximate} for Hamming distance, and later was extended to $l_p$ distance ($0 < p \leq 2$) by Datar et al. \cite{datar2004locality}, which led to E2LSH \cite{andoni2005e2lsh} for Euclidean distance ($p=2$). E2LSH adopts this framework as follows. In the indexing phase, E2LSH concatenates $K$ i.i.d. LSH functions $h_1,h_2,\cdots,h_K$ to form a compound hash function $G$, i.e., $G(o) = (h_1(o),h_2(o),\cdots,h_K(o))$ for all $o \in \mathcal{D}$. If two objects $o$ and $q$ have the same hash value, i.e., $G(o) = G(q)$, we say $o$ and $q$ \emph{collide} in the same bucket under $G$. E2LSH samples uniformly at random $L$ such hash functions $G_1(\cdot),G_2(\cdot),\cdots,G_L(\cdot)$ and builds $L$ hash tables. In the query phase, E2LSH computes $L$ hash values $G_1(q),G_2(q),\cdots,G_L(q)$ and lookups the corresponding $L$ buckets to find the candidates of $q$. The variants of E2LSH, such as LSH-Forest \cite{bawa2005lsh}, Multi-Probe LSH \cite{lv2007multi}, LSB-Forest \cite{tao2009quality}, and SK-LSH \cite{liu2014sk}, follow this search framework.

Notably, E2LSH conducts the $c$-ANNS with sublinear time $O(dn^\rho\log_{1/p_2}(n))$, where $\rho = \ln(1/p_1) / \ln(1/p_2)$ \cite{datar2004locality}. The reason is that the static concatenating search framework can effectively avoid the false positives in the sense that the far-apart objects hardly collide with $q$. Due to the use of $K$ concatenated LSH functions, such negative probability decreases significantly from $p_2$ to $p_2^K$. However, the positive probability also decreases significantly from $p_1$ to $p_1^K$, and hence the true positives are not easy to be identified neither. For example, as shown in Figure \ref{fig:example:e2lsh}, suppose $o_1$ is the NN of $q$, $o_2$ is also close to $q$, while $o_3$ is far-apart from $q$. We consider $K=4$ and $L=2$. Due to the use of this framework, $o_3$ does not collide with $q$, but the close objects $o_1$ and $o_2$ also fail to collide with $q$. To achieve a certain recall, the number of hash tables (i.e., $L$) of E2LSH is often set to be more than one hundred, and sometimes up to several hundred \cite{gan2012locality}, leading to a large amount of indexing overhead. 

\subparagraph{Dynamic Collision Counting Framework}
To reduce the large indexing overhead, Gan et al. \cite{gan2012locality} introduced a dynamic collision counting framework and the C2LSH scheme accordingly. In the indexing phase, C2LSH uses $m$ independent LSH functions $h_1,h_2,\cdots,h_m$ to build $m$ hash tables individually. Two objects $o$ and $q$ \emph{collide} in the same bucket under $h$ if $h(o) = h(q)$. The idea of C2LSH is that, if $o$ is close to $q$ in the original space $\mathbb{R}^d$, then $o$ and $q$ will collide \emph{frequently} among the $m$ hash tables. Thus, in the query phase, C2LSH maintains the collision number $\#Col(o)$ for each $o$ which collides with $q$, and $o$ is considered as an NN candidate of $q$ if $\#Col(o) \geq l$, where $l$ is the collision threshold. C2LSH returns the final answers from a set of such candidates. In fact, this framework can be considered as a \emph{dynamic} $l$-concatenating search framework, because it checks a candidate $o$ until $\#Col(o) \geq l$. Compared to the static concatenating search framework which uses $KL$ LSH functions to generate $L$ combinations only, this framework can generate $\tbinom{m}{l}$ combinations for each $o$. 
Thus, for the same recall, C2LSH requires much less number of LSH functions than E2LSH, and thus takes much less indexing overhead. Various extensions, such as QALSH \cite{huang2015query,huang2017query} and LazyLSH \cite{zheng2016lazylsh}, are proposed based on this framework.

However, the query time complexity of C2LSH in the worst case is $O(n\log n)$ \cite{gan2012locality}, which limits its scalability for large $n$. Notice that C2LSH builds hash tables for every single LSH function. Even though $p_2$ is small for the far-apart objects, there are expected $(1-(1-p_2)^m)n\approx p_2 m n$ objects with at least one collision, which cannot be neglected especially for large $n$. 
For example, as shown in Figure \ref{fig:example:c2lsh}, suppose $m=8$ and $l=4$, C2LSH can identify the close objects $o_1$ and $o_2$ since $\#Col(o_1) = \#Col(o_2) > l$, but it also conducts 3 times collision counting for the far-apart object $o_3$. 

\vspace{-0.15em}
\paragraph{Our Method}
To achieve a better trade-off between space and query time, we introduce a novel LSH scheme based on the \emph{Longest Circular Co-Substring} (LCCS) search framework (LCCS-LSH). We first introduce a novel concept of LCCS and a new data structure named Circular Shift Array (CSA) for $k$-LCCS search. Then, in the indexing phase, we exploit a collection of $m$ independent LSH functions $h_1,h_2,\cdots,h_m$ to convert data objects into hash strings of length $m$, i.e., $H(o) = [h_1(o),h_2(o),\cdots,h_m(o)]$. The insight is that, if $o$ is close to $q$ in $\mathbb{R}^d$, then $H(o)$ will have a longer LCCS with $H(q)$ than the hash strings for the far-apart ones with high probability. Let $|LCCS(H(o),H(q))|$ be the length of LCCS between $H(o)$ and $H(q)$. In the query phase, we find data objects with the largest $|LCCS(H(o),H(q))|$ as candidates of $q$ and get the final answers from a set of such candidates. 
For example, as shown in Figure \ref{fig:example:lcslsh}, suppose $m=8$ and we combine the $8$ hash values for each object as a circular hash string. $|LCCS(H(o_1),H(q))|=5$, which is larger than $|LCCS(H(o_2),H(q))|$ and $|LCCS(H(o_3),H(q))|$, which are $3$ and $2$, respectively.
Thus, the NN $o_1$ can be determined efficiently. Furthermore, since LCCS-LSH works on the hash strings only, which is independent of data types, it is \emph{LSH-family-independent} and can be applied to $c$-ANNS under different distance metrics that admit LSH families.

\vspace{-0.15em}
\paragraph{Contributions}
In this paper, we introduce a novel LSH scheme LCCS-LSH for high-dimensional $c$-ANNS. The LCCS search framework dynamically concatenates consecutive hash values for data objects, which can identify the close objects in an efficient and effective manner and it requires to tune only a single parameter $m$. LCCS-LSH enjoys a quality guarantee on query results, and we further analyse its space and time complexities. In addition, we introduce a multi-probe version of LCCS-LSH to reduce the indexing overhead. Experimental results over five real-life datasets demonstrate that LCCS-LSH outperforms state-of-the-art LSH schemes, such as Multi-Probe LSH and FALCONN. 

\paragraph{Organization}
The roadmap of the paper is as follows. 
Section \ref{sect:preliminary} discusses the problem settings.
The LCCS search framework is introduced in Section \ref{sect:framework}. LCCS-LSH and its theoretical analysis are presented in Sections \ref{sect:methods} and \ref{sect:analysis}, respectively. Section \ref{sect:expt} reports experimental results. Section \ref{sect:related_work} surveys the related work. Finally, we conclude our work in Section \ref{sect:conclusion}.

\section{Preliminaries}
\label{sect:preliminary}
Before we introduce the LCCS-LSH scheme, we first review some preliminary knowledge.
\vspace{-0.5em}

\subsection{Problem Settings}
\label{sect:preliminary:problem}
In this paper, we consider data objects and queries represented as vectors in $d$-dimensional space $\mathbb{R}^d$. Let $Dist(o,q)$ be a distance metric between any two objects $o$ and $q$. Suppose $\mathcal{D}$ is a database of $n$ data objects from $\mathbb{R}^d$. Given a query $q$, we say $o^*$ is the Nearest Neighbor (NN) of $q$ such that $o^* = \arg\min_{o\in \mathcal{D}} Dist(o,q)$. Then, 
\begin{definition}[$c$-ANNS]
Given an approximation ratio $c$ ($c > 1$), the problem of $c$-ANNS is to construct a data structure which, for any query $q \in \mathbb{R}^d$, finds a data object $o \in \mathcal{D}$ such that $Dist(o,q) \leq c \cdot Dist(o^*,q)$, where $o^* \in \mathcal{D}$ is the NN of $q$.
\end{definition} 

Similarly, the problem of $c$-$k$-ANNS is to construct a data structure which, for any query $q \in \mathbb{R}^d$, finds $k$ data objects $o_i \in \mathcal{D}$ ($1 \leq i \leq k$) such that $Dist(o_i,q) \leq c \cdot Dist(o_i^*,q)$, where $o_i^* \in \mathcal{D}$ is the $i^{th}$ NN of $q$. 

LSH schemes \cite{indyk1998approximate,charikar2002similarity,datar2004locality,andoni2006near} cannot solve the problem of $c$-ANNS directly. Instead, they solve the problem of $(R,c)$-Near Neighbor Search ($(R,c)$-NNS), which is a decision version of $c$-ANNS. One can reduce the $c$-ANNS problem to a series of $(R,c)$-NNS via a binary-search-like method within a log factor overhead, where $R \in \bracket{1,c,c^2,\cdots}$. Formally, 
\begin{definition}[$(R,c)$-NNS]
Given a search radius $R$ ($R > 0$) and an approximation ratio $c$ ($c > 1$), the problem of $(R,c)$-NNS is to construct a data structure which, for any $q \in \mathbb{R}^d$, returns objects that satisfy the following conditions: 
\begin{itemize}
\item If there is an object $o \in \mathcal{D}$ such that $Dist(o,q) \leq R$, then return an arbitrary object $o^\prime$ such that $Dist(o^\prime,q) \leq cR$;
\item If $Dist(o,q) > cR$ for all $o \in \mathcal{D}$, then return nothing;
\item Otherwise, the result is undefined.
\end{itemize}
\end{definition}

LCCS-LSH is orthogonal to the LSH family and can handle various kinds of distance metrics. Thus, $Dist(\cdot,\cdot)$ can be the widespread distance metrics, such as Euclidean distance, Hamming distance, Angular distance, and so on. In this paper, we focus on two popular distance metrics, i.e., Euclidean distance and Angular distance, to demonstrate the superior performance of LCCS-LSH. Notice that we do \emph{not} claim that every distance metric can be handled by LCCS-LSH. It supports the distance metrics if and only if there exist LSH families for them. 

\subsection{Locality-Sensitive Hashing}
\label{sect:preliminary:LSH}
LSH schemes \cite{indyk1998approximate, charikar2002similarity, datar2004locality, andoni2006near, terasawa2007spherical, har2012approximate, andoni2015optimal} are one of the most popular methods for $c$-ANNS. Given a hash function $h$, we say two objects $o$ and $q$ collide in the same bucket if $h(o) = h(q)$. Formally, an LSH family is defined as follows \cite{har2012approximate}.
\begin{definition}[LSH Family] 
\label{def:lsh_family}
Given a search radius $R$ ($R > 0$) and an approximation ratio $c$, a hash family $\mathcal{H} = \{h:\mathbb{R}^d \rightarrow \mathbb{U}\}$ is said to be $(R,cR,p_1,p_2)$-sensitive, if for any $o,q \in \mathbb{R}^d$, $\mathcal{H}$ satisfies the following conditions:
\begin{itemize}
\item If $Dist(o,q) \leq R$,  then $\Pr_{h \in \mathcal{H}} [h(o) = h(q)] \geq p_1$;
\item If $Dist(o,q) > cR$, then $\Pr_{h \in \mathcal{H}} [h(o) = h(q)] \leq p_2$;
\item $c > 1$ and $p_1 > p_2$.
\end{itemize}
\end{definition}

With an LSH family $\mathcal{H}$, we have Theorem \ref{theorem:time-complexity} for the static concatenating search framework as follows \cite{har2012approximate}.
\begin{theorem}[Theorem 3.4 in \cite{har2012approximate}]
\label{theorem:time-complexity}
Given an $(R,cR,p_1,p_2)$-sensitive hash family $\mathcal{H}$, one can build a data structure for the $(R,c)$-NNS which uses $O(n^{1+\rho}/p_1)$ space and $O(dn^\rho/p_1 \cdot \lceil \log_{1/{p_2}}(n) \rceil)$ query time, where $\rho=\ln(1/p_1)/\ln(1/p_2)$.
\end{theorem}

Next, we review two LSH families, i.e., the random projection LSH family \cite{datar2004locality} and the cross polytope LSH family \cite{terasawa2007spherical}, for Euclidean distance and Angular distance, respectively.

\paragraph{Random Projection LSH Family}
The random projection LSH family \cite{datar2004locality} is designed for Euclidean distance. Given two objects $o=(o_1,o_2,\cdots,o_d)$ and $q=(q_1,q_2,\cdots,q_d)$, Euclidean distance is computed as $\norm{o-q}=\sqrt{\sum_{i=1}^d (o_i - q_i)^2}$. The LSH function is defined as follows:
\begin{equation}
\label{eqn:l2-lsh-func}
h_{\vec{a},b}(o)=\floor{\frac{\vec{a} \cdot \vec{o} + b}{w}}, 
\end{equation}
where $a$ is a $d$-dimensional vector with each entry chosen i.i.d from standard Gaussian distribution $\mathcal{N}(0, 1)$; $w$ is a pre-specified bucket width; $b$ is a random offset chosen uniformly at random from $[0,w)$. 

Given any two objects $o, q \in \mathbb{R}^d$, let $\tau = \norm{o-q}$. The collision probability $p(\tau)$ is computed as follows \cite{datar2004locality}: 
\begin{equation}
\label{eqn:l2-col-prob}
\begin{array}{rcl}
p(\tau)
&=& \Pr[h_{\vec{a},b}(o)=h_{\vec{a},b}(q)] \\
&=& 1 - 2\Phi(-w/\tau) - \frac{2}{\sqrt{2\pi}(w/\tau)}(1 - e^{-{(w/\tau)}^2/2}),
\end{array}
\end{equation}
where $\Phi(x) = \int_{-\infty}^x \frac{1}{\sqrt{2\pi}} e^{-x^2/2}\, dx$ is the Cumulative Distribution Function (CDF) of $\mathcal{N}(0, 1)$.

\paragraph{Cross Polytope LSH Family}
Let $\mathcal{S}^{d-1}$ be the unit sphere in $\mathbb{R}^d$ centered in the origin. The cross polytope LSH family \cite{terasawa2007spherical} is designed for the Euclidean distance on $\mathcal{S}^{d-1}$, which is equivalent to the Angular distance. Given two objects $o=(o_1,o_2,\cdots,o_d)$ and $q=(q_1,q_2,\cdots,q_d)$, Angular distance is computed as $\theta(o,q) = cos^{-1}(\frac{\vec{o} \cdot \vec{q}}{\norm{o} \norm{q}})$. 

The cross polytope LSH family has been shown to outperform the hyperplane LSH family \cite{charikar2002similarity} and achieves the asymptotically optimal hash quality $\rho$ \cite{terasawa2007spherical,andoni2015practical}. Let $A \in \mathbb{R}^{d \times d}$ be a random rotation matrix with each entry drawn i.i.d from $\mathcal{N}(0, 1)$. Suppose $e_i$ is the $i^{th}$ standard basis vector of $\mathbb{R}^d$ and $u_j \in \{\pm e_i\}_{1 \leq i \leq d}$. Given any object $o \in \mathcal{S}^{d-1}$, i.e., $\norm{o}=1$, the LSH function is defined as follows:
\begin{equation}
\label{eqn:angular-lsh-func}
h_{A}(o) = \arg\min_j \norm{u_j - A\cdot o/\norm{A\cdot o} }.
\end{equation}

Given any two objects $o, q \in \mathcal{S}^{d-1}$, let $\tau = \norm{o-q}$, where $0 < \tau < 2 $. The collision probability $p(\tau)$ can be computed as follows \cite{andoni2015practical}:
\begin{equation}
\label{eqn:angular-col-prob}
\ln \frac{1}{p(\tau)} = \frac{\tau^2}{4-\tau^2} \cdot \ln d + O_{\tau}(\ln \ln d),
\end{equation}
and the hash quality $\rho$ can be computed as follows \cite{andoni2015practical}:
\begin{equation}
\label{eqn:angular-rho}
\rho = \frac{1}{c^2} \cdot \frac{4-c^2 R^2}{4 - R^2} + o(1).
\end{equation}

\section{The LCCS Search Framework}
\label{sect:framework}
In this section, we present the LCCS search framework. We introduce the concepts of LCCS and $k$-LCCS search in Section \ref{sect:framework:def}. Then, we propose a novel data structure Circular Shift Array (CSA) for $k$-LCCS search in Section \ref{sect:framework:k-lccs}. 

\subsection{Definition of LCCS} 
\label{sect:framework:def}
We first introduce the definition of \emph{Circular Co-Substring}. It can be considered as the common circular substring of two strings starting from the same position. Formally, 
\begin{definition}
\label{def:ccs}
Given two strings $T=[t_1,t_2,\cdots,t_m]$ and $Q=[q_1,q_2,\cdots,q_m]$ of the same length $m$, a string $X$ is a Circular Co-Substring of $T$ and $Q$ if and only if $X$ is an empty string, $X=[t_i,t_{i+1},\cdots,t_j]=[q_i,q_{i+1},\cdots,q_j]$, or $X=[t_j,\cdots,t_m,t_1,\cdots,t_i]=[q_j,\cdots,q_m,$ $q_1,\cdots,q_i]$, where $1 \leq i < j \leq m$. 
\vspace{-0.5em}
\end{definition}

\begin{algorithm*}[t]
\caption{Building CSA}
\label{alg:indexing}
\KwIn{$\mathcal{T}$: a dataset of $n$ strings of length $m$ such that $\mathcal{T} = \{T_1,T_2,\cdots,T_n\}$ and $|T_i| = m$.}
\KwOut{$m$ sorted indices $\{I_1,I_2,\cdots,I_m\}$ and $m$ next links $\{N_1, N_2, \cdots, N_m\}$.}
\For {$i = 1$ to $m$} {
	$I_i=\arg sort(shift(\mathcal{T}, i-1))$\Comment*[r]{$I_i$ is the sorted index of $n$ strings in $shift(\mathcal{T}, i-1)$}
}
\For {$i = 1$ to $m$} {
	\For {$j = 1$ to $n$}{
		$pos[I_{i\%m + 1}[j]] = j$\Comment*[r]{$pos$ is the position of $n$ strings in the next sorted index $I_{i\%m + 1}$}
	}
	\For {$j = 1$ to $n$}{
		$N_i[j] = pos[I_i[j]]$\Comment*[r]{$N_i$ is the position of $n$ strings in the next sorted $shift(\mathcal{T},i\%m)$}
	}
}
\Return $\{I_1,I_2,\cdots,I_m\}$ and $\{N_1, N_2, \cdots, N_m\}$\;
\end{algorithm*}

\begin{example}
\label{exmp:ccs}
Consider two strings $T = [1,2,3,4,1,5]$ and $Q = [1,1,2,3,4,5]$ as an example. The substring $[5,1]$ is a Circular Co-Substring of $T$ and $Q$. However, although the substring $[1,2,3,4]$ is a common circular substring of $T$ and $Q$, it is not a Circular Co-Substring, because it does not start from the same position of $T$ and $Q$.
\hfill $\triangle$ \par 
\end{example}

Let $|T|$ be the length of a string $T$. The Longest Circular Co-Substring (LCCS) is defined as follows.
\begin{definition}
\label{def:lccs}
Given any two strings $T$ and $Q$ of the same length, let $\mathcal{S}(T,Q)$ be the set of all Circular Co-Substrings of $T$ and $Q$. The LCCS of $T$ and $Q$ is defined as $LCCS(T,Q)=\arg\max_{X \in \mathcal{S}(T,Q)} |X|$.
\end{definition}

The problem of $k$-Longest Circular Co-Substring search ($k$-LCCS search) is defined as follows.
\begin{definition}
\label{def:k-lccs}
Given a collection of strings $\mathcal{T}$ of the same length $m$, the problem of $k$-LCCS search is to construct a data structure which, for any query string $Q$ that $|Q|=m$, finds a set of strings $\mathcal{T}^* \subseteq \mathcal{T}$ with cardinality $k$ such that for all $T^* \in \mathcal{T}^*, T' \in \mathcal{T} \backslash \mathcal{T}^*$, $|LCCS(T',Q)| \leq |LCCS(T^*,Q)|$.
\end{definition}

\subsection{$k$-LCCS Search}
\label{sect:framework:k-lccs}
Suppose $LCP(T,Q)$ is the Longest Common Prefix (LCP) between two strings $T$ and $Q$. Given a string $T=[t_1,t_2,\cdots,t_m]$ and an integer $i \in \{0,1,\cdots,m-1\}$, let $shift(T,i)=[t_{i+1},$ $\cdots,t_m,t_1,\cdots,t_{i}]$ be the circular string of $T$ after shifting $i$ positions. Since the index of $T$ starts from $1$, $shift(T, i-1)$ corresponds to the circular string of $T$ starting from $t_{i}$. For simplicity, for a collection of strings $\mathcal{T}$, we let $shift(\mathcal{T},i) = \{shift(T,i) \mid T \in \mathcal{T}\}$. 

To solve the problem of $k$-LCCS search, we propose a data structure named \emph{Circular Shift Array (CSA)}, which is inspired by Suffix Array \cite{manber1993suffix}. Specifically, the insight of CSA comes from Fact \ref{fact:lcp}, which is described as follows.
\begin{fact}
\label{fact:lcp}
Given two strings $T$ and $Q$ of the same length $m$, $LCCS(T$, $Q) = \max_{i \in \{0,1,\cdots,m-1\}} LCP(shift(T,i), shift(Q,i))$.
\end{fact}

According to Fact \ref{fact:lcp}, the $LCCS(T,Q)$ can be identified by considering the LCP of all shifted $T$'s and $Q$'s. Let $\prec$ and $\preceq$ be the alphabetical order relationships of strings, where $\preceq$ allows equal cases. We have Fact \ref{fact:prefix} as follows.
\begin{fact}
\label{fact:prefix}
If $T_1 \preceq T_2 \prec T_3$, then $\forall Q$, 
\begin{displaymath}
|LCP(T_2, Q)| \geq \min(|LCP(T_1, Q)|, |LCP(T_3, Q)|).
\end{displaymath}
\end{fact}

The soundness of Fact \ref{fact:prefix} is obvious. Let $\min(\mathcal{T})$ and $\max(\mathcal{T})$ be the minimum and maximum string in alphabetical order of $\mathcal{T}$, respectively. Then,
\begin{corollary}
\label{corollary:prefix}
For any string $Q$ s.t. $\min(\mathcal{T}) \preceq Q \prec \max(\mathcal{T})$, let $T_l=\arg\max_{T \in \mathcal{T}} T \preceq Q$ and $T_u=\arg\min_{T \in \mathcal{T}} Q \prec T$ be the lower bound and upper bound of $Q$, respectively. If $T^* = \arg\max_{T \in \mathcal{T}} |LCP(T, Q)|$, then $T^*=T_l$ or $T^*=T_u$. 
\end{corollary}

Corollary \ref{corollary:prefix} indicates that, given a query string $Q$ and a database of $n$ sorted strings $\mathcal{T}$ in alphabetical order, one can use binary search on $\mathcal{T}$ to find $T^*$ in $O(m+\log n)$ time. It yields a simple method with two phases to answer the $1$-LCCS query as follows:
In the indexing phase, given a database $\mathcal{T}$ of $n$ strings of length $m$, we sort $shift(\mathcal{T},i-1)$ in alphabetical order and maintain the sorted index $I_i$ for each $i \in \{1,2,\cdots,m\}$. In the query phase, to find the $1$-LCCS of $Q$, we conduct binary search on each sorted index $I_i$ to get $T_i^*$ such that $T_i^* = \arg\max_{T \in shift(\mathcal{T},i-1)} |LCP(T, shift(Q,i-1))|$; the $1$-LCCS of $Q$ is the string $T^*$ among $\{T_1^*,T_2^*,\cdots,T_m^*\}$ with the largest $|LCP(T_i^*,shift(Q,i-1))|$.

This simple method requires $m$ times binary search, and hence the query time complexity is $O(m (m+\log n))$. Next, we introduce a strategy to reduce the query time complexity to $O(m+\log n)$ under certain assumptions. 

\begin{algorithm*}[t]
\caption{$k$-LCCS Search using CSA}
\label{alg:query}
\KwIn{ $\mathcal{T} = \{T_1,T_2,\cdots,T_n\}$, $\{I_1,I_2,\cdots,I_m\}$, $\{N_1,N_2,\cdots,N_m\}$, $Q$, and \#candidates $k$;}
\KwOut{$\mathcal{C}$: the results of $k$-LCCS search.}
$\mathcal{C} \gets \emptyset$; $PQ \gets \emptyset$ \Comment*[r]{$\mathcal{C}$ is a candidate set and $PQ$ is a priority queue}
(${pos}_{l,1}, {pos}_{u,1}, {len}_{l,1}, {len}_{u,1}$) $\gets$ BinarySearch($I_1,Q$) \Comment*[r]{${pos}_{l,i}$ and ${pos}_{u,i}$ are the positions of $T_{l,i}$ and $T_{u,i}$ in $I_i$}
$PQ.push({len}_{l,1}, {pos}_{l,1}, 1, -1)$ \Comment*[r]{$-1$ represents the down direction}
$PQ.push({len}_{u,1}, {pos}_{u,1}, 1, +1)$ \Comment*[r]{$+1$ represents the up direction} 
\For {$i = 2$ to $m$} {
	\If{${len}_{l,i-1} \geq 1$ and ${len}_{u,i-1} \geq 1$} {
		(${pos}_{l,i}, {pos}_{u,i}, {len}_{l,i}, {len}_{u,i}$) $\gets$ BinarySearchBetween($I_i, shift(Q,i-1), N_{i-1}[I_{i-1}[{pos}_{l,i-1}]], N_{i-1}[I_{i-1}[{pos}_{u,i-1}]]$)\;
	} 
	\Else {
		(${pos}_{l,i}, {pos}_{u,i}, {len}_{l,i}, {len}_{u,i}$) $\gets$ BinarySearch($I_i, shift(Q,i-1)$)\;
	}
	$PQ.push({len}_{l,i}, {pos}_{l,i}, i, -1)$\; 
	$PQ.push({len}_{u,i}, {pos}_{u,i}, i, +1)$\; 
}
\While{$|\mathcal{C}| < k$} {
	$(len, pos, i, dir) \gets PQ.top()$; $PQ.pop()$\;
	$\mathcal{C} \gets \mathcal{C} \cup \{I_i[pos]\}$\;
	$PQ.push(|LCP(shift(T_{I_i[pos]},i-1), shift(Q,i-1))|, pos+dir, i, dir)$\;
}
\Return $\mathcal{C}$\;
\end{algorithm*}

\begin{lemma}
\label{lemma:tail}
Suppose $T_l \preceq Q \prec T_u$. For any $k \geq 1$, if we have $|LCP(T_l,Q)|\geq k$ and $|LCP(T_u,Q)| \geq k$, then $shift(T_l, k) \preceq shift(Q, k) \prec shift(T_u, k)$.
\end{lemma}

Lemma \ref{lemma:tail} is true according to the definition of $\prec$ and $\preceq$. According to Lemma \ref{lemma:tail}, once we conduct binary search on $shift(\mathcal{T},i-1)$ for a query string $shift(Q,i-1)$ and find $T_{l,i}$ and $T_{u,i}$ as its lower bound and upper bound, respectively, we can immediately know a loose lower bound and a loose upper bound for $shift(Q, i)$. Let ${len}_{l,i}=|LCP(T_{l,i},shift(Q,i-1))|$ and ${len}_{u,i}=|LCP(T_{u,i}$, $shift(Q,i-1))|$. Then,
\begin{corollary}
\label{corollary:tail-cover}
If ${len}_{l,i} \geq 1$ and ${len}_{u,i} \geq 1$, then the lower bound $T_{l,i+1}$ and upper bound $T_{u,i+1}$ of $shift(Q,i)$ satisfy that $shift(T_{l,i}, 1) \preceq T_{l,i+1} \preceq shift(Q, i) \prec T_{u,i+1} \preceq shift(T_{u,i}, 1)$. 
\end{corollary}

Based on Lemma \ref{lemma:tail} and Corollary \ref{corollary:tail-cover}, the simple method discussed before can be further optimized. 
To find the $1$-LCCS of $Q$, we conduct only once binary search on the whole $shift(\mathcal{T},0)$ ($i=1$) for the query $shift(Q,0)$ (or simply $Q$). 
Then, for the query $shift(Q,i-1)$ ($i>1$), according to Corollary \ref{corollary:tail-cover}, we can conduct binary search on $shift(\mathcal{T},i-1)$ between $shift(T_{l,i-1},1)$ and $shift(T_{u,i-1},1)$. After that, we get $T_{l,i}$ and $T_{u,i}$ and can continue to use them to narrow down the binary search range for the next query $shift(Q,i)$. We repeat this procedure, and the $1$-LCCS of $Q$ will be the string with the longest LCP among $LCP(T_{l,i}, shift(Q,i-1))$ and $LCP(T_{u,i}, shift(Q,i-1))$ for all $i \in \{1,2,\cdots,m\}$. 

To speed up the query phase, we need to know the positions of $shift(T_{l,i},1)$ and $shift(T_{u,i},1)$ when we get $T_{l,i}$ and $T_{u,i}$. Thus, in the indexing phase, we not only need to maintain the sorted indices $\{I_1,I_2,\cdots,I_m\}$, but also require to store the next links $\{N_1,N_2,\cdots,N_m\}$, e.g., $N_i$ stores the positions of $\mathcal{T}$ in the next sorted $shift(\mathcal{T},i\%m)$. The pseudo-code of building CSA is depicted in Algorithm \ref{alg:indexing}.

To find the $k$-LCCS of $Q$, we first follow the procedure of $1$-LCCS search and compute $T_{l,i}$ and $T_{u,i}$ for each $shift(Q,i-1)$. Then, we construct a priority queue $PQ$ and perform a $2m$-way sorted list merge. The strings with top-$k$ longest lengths in $PQ$ are the $k$-LCCS results of $Q$. The pseudo-code of $k$-LCCS search is shown in Algorithm \ref{alg:query}. 

\begin{figure}[tb]
\centering   
\includegraphics[width=0.45\textwidth]{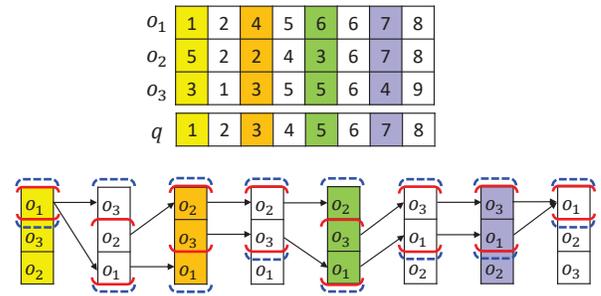}
\vspace{-0.5em}
\caption{An example of the $1$-LCCS search}   
\vspace{-1.0em}
\label{fig:lsh-example}
\end{figure}

\begin{example}
\label{exmp:lccs-lsh}
We now use an example to illustrate Algorithm \ref{alg:query}. Suppose $k=1$. We continue to use the same $o_1,o_2,o_3$ and $q$ from Figure \ref{fig:example} as in Figure \ref{fig:lsh-example}. We follow Algorithm \ref{alg:indexing} and build CSA with sorted indices $\{I_1,I_2,\cdots,I_8\}$ and next links $\{N_1,N_2,\cdots$, $N_8\}$, e.g., $I_1 = [1,3,2]$ and $N_1=[3,1,2]$. 

Given a query string $q=[1,2,3,4,5,6,7,8]$, to find the $1$-LCCS of $q$, we first conduct binary search on the whole $I_1$, and get the positions of $T_{l,1}$ and $T_{u,1}$, i.e., ${pos}_{l,1}$ and ${pos}_{u,1}$, as depicted by red brackets (line 2). Since $q \prec o_1$, ${pos}_{l,1}={pos}_{u,1}=1$. Then, along with $N_i$, the binary search range on the next $I_{i+1}$ can be determined, as shown by blue dash brackets (lines 5--9). For example, consider $I_5$, since $shift(o_3$, $4) \prec shift(q,4) \prec shift(o_1,4)$, ${pos}_{l,5}=2$ and ${pos}_{u,5}=3$. Since $N_5[{pos}_{l,5}]=1$ and $N_5[{pos}_{u,5}]=2$, the binary search range on $I_6$ is narrowed down to $[1,2]$. We use a priority queue to check the objects with longest LCP among $\{I_1,I_2$, $\cdots,I_8\}$ (lines 3--4, lines 8--9, and lines 12--15). $o_1$ is first verified because $o_1$ has the largest $|LCP(shift(o_1,5),shift(q,5))| = 5$ on $I_6$, and $o_1$ is the $1$-LCCS of $q$.
\hfill $\triangle$ \par 
\end{example}

\begin{theorem}
\label{theorem:klccs}
Let $T=[t_1,t_2,\cdots,t_m]$ and $Q=[q_1,q_2,\cdots$, $q_m]$. If the probability that $t_i=q_i$ equals to $p$ and is independent for each $i$, one can build a data structure CSA using Algorithm \ref{alg:indexing} with $O(nm)$ space and $O(mn\log n)$ time, and answer the $k$-LCCS queries using Algorithm \ref{alg:query} within $O(\log n + (m+k)\log m)$ time.
\end{theorem}

\begin{proof}
The space complexity of CSA is obvious. Since $\{I_1,I_2,\cdots,I_m\}$ and $\{N_1,N_2,\cdots,N_m\}$ require $O(nm)$ space, the space complexity of Algorithm \ref{alg:indexing} is also $O(nm)$. Algorithm \ref{alg:indexing} requires $m$ times quick sort and each takes $O(n \log n)$ time. Thus, the indexing time complexity is $O(mn\log n)$. 

For the $k$-LCCS search, Algorithm \ref{alg:query} first conducts binary search on $I_1$, which requires $O(\log n)$ time (line 2). At each iteration $i$ (lines 5-11), there are expected $O(1/p)$ objects between the lower bound $N_{i-1}[I_{i-1}[{pos}_{l,i-1}]]$ and upper bound $N_{i-1}[I_{i-1}[{pos}_{u,i-1}]]$; since $p$ is a constant value, each binary search takes $O(\log(\min(1/p, n)))=O(1)$ time only. To find the $k$-LCCS of $Q$, there are $O(m+k)$ priority queue operations on average (lines 12-15), and each takes at most $O(\log m)$ time. Thus, the time complexity of Algorithm \ref{alg:query} is $O(\log n + (m+k)\log m)$.
\end{proof}

\section{The LCCS-LSH Scheme}
\label{sect:methods}
In this section, we present the LCCS-LSH schemes for high-dimensional $c$-ANNS. Section \ref{sect:methods:single} introduces the single-probe version of LCCS-LSH. We design a heuristic multi-probe version of LCCS-LSH in Section \ref{sect:methods:multi}. 
\vspace{-0.5em}

\subsection{Single-Probe LCCS-LSH}
\label{sect:methods:single}
The single-probe LCCS-LSH scheme (or simply LCCS-LSH) consists of two phases: indexing phase and query phase. 

\paragraph{Indexing Phase}
Given a database $\mathcal{D}$ of $n$ data objects, LCCS-LSH first generates $m$ i.i.d. LSH functions $h_1,h_2,\cdots,h_m$ from the LSH family $\mathcal{H}$. Then, it computes the $m$ hash values $h_1(o),h_2(o),\cdots,h_m(o)$ for each $o \in \mathcal{D}$ and concatenates all of them to a hash string $H(o) = [h_1(o),h_2(o)$, $\cdots,h_m(o)]$ of length $m$. Let $\mathcal{T}=\{H(o) \mid o \in \mathcal{D}\}$ be a collection of $n$ such hash strings. Finally, LCCS-LSH constructs a data structure CSA for $\mathcal{T}$ using Algorithm \ref{alg:indexing}.

\paragraph{Query Phase}
For the $c$-ANNS of $q$, LCCS-LSH first computes the hash string $H(q)$. Then, it conducts a $\lambda$-LCCS search of $H(q)$ using Algorithm \ref{alg:query}, and gets a set $\mathcal{C}$ of candidates such that $|\mathcal{C}|=\lambda$. Finally, we compute the actual distance between each candidate $o \in \mathcal{C}$ and $q$, and return the nearest one as the $c$-ANNS answer of $q$. For the $c$-$k$-ANNS of $q$, LCCS-LSH only needs to conduct $(\lambda+k-1)$-LCCS search of $H(q)$ and verifies $(\lambda+k-1)$ candidates from $\mathcal{C}$ accordingly. The nearest $k$ objects among $\mathcal{C}$ are the $c$-$k$-ANNS answers of $q$. $\lambda$  is a parameter which is determined by $m$ and $n$. We will discuss the settings of $m$ and $\lambda$ in Section \ref{sect:analysis}.

Notably, the $LCCS(H(o),H(q))$ between $H(o)$ and $H(q)$ can be considered as a dynamic concatenation of $l$ consecutive hash values, i.e., $h_i(o),h_{i+1}(o),\cdots,h_{(i+l)\% m}(o)$, where $l = |LCCS(H(o),H(q))|$. Thus, the LCCS search framework can be considered as a \emph{dynamic concatenating search framework}. Similar to the static concatenating search framework, the false positives can be effectively avoided due to the concatenation. Furthermore, since Algorithm \ref{alg:query} prioritizes the objects with the largest $|LCCS(H(o),H(q))|$ as candidates, it can also identify the correct answers efficiently.

\subsection{Multi-Probe LCCS-LSH}
\label{sect:methods:multi}
The multi-probe schemes are widely used to reduce space overhead, such as Multi-Probe LSH \cite{lv2007multi} for random projection LSH family \cite{datar2004locality} and FALCONN \cite{andoni2015practical} for cross-polytope LSH family \cite{terasawa2007spherical}. However,
they are designed for the static concatenating search framework. It is inefficient to trivially adapt existing multi-probe schemes to LCCS-LSH.

\paragraph{Challenges}
To explain the reason why existing multi-probe schemes do not work well with LCCS-LSH, we first consider a trivial multi-probe extension: 
given a hash string $H(q)=[h_1(q),h_2(q),\cdots,h_m(q)]$, we adopt existing multi-probe schemes to (virtually) generate a sequence of probes by modifying some of $h_i(q)$ among $H(q)$; then, we conduct a $\lambda$-LCCS search of this modified $H(q)$ in the probing sequence using Algorithm \ref{alg:query}. 
This trivial multi-probe extension, however, has two major problems. 
Firstly, if we modify a single $h_i$ only, since Algorithm \ref{alg:query} uses LCP to find the objects with largest $|LCCS(H(o),H(q))|$, the LCP from most of the positions after $i$, i.e., $i+1,i+2,\cdots$, are identical to those before modification, which should be avoided. 
Secondly, for the $\lambda$-LCCS search of $H(q)$ after two modifications which are far away from each other, it is very likely that the new probed objects were checked in previous probing sequence, leading to redundant computations.

\vspace{-0.5em}
\begin{example}
We now use Figure \ref{fig:mp-lccs-example} to illustrate these two problems. Suppose $H(q)=[1,2,3,4,5,6,7,8]$ and $H^{(1)}$, $H^{(2)}$, and $H^{(3)}$ are three alternative probes by modifying $h_1(q)$ and $h_4(q)$ to $5$. Let $\mathcal{T}=\{H(o_1), H(o_2), H(o_3)\}$. 
Firstly, by modifying $H(q)$ to $H^{(1)}$, except for $shift(\mathcal{T}, 2)$ and $shift(\mathcal{T}, 3)$, the objects that have longest LCP of $H^{(1)}$ from other $shift(\mathcal{T}, i)$ do not change, i.e., $i \in \{4,5,6,7,0,1\}$. Thus, they should be avoided for the $\lambda$-LCCS search of $H^{(1)}$. Similarly, for $H^{(2)}$, we only need to consider $shift(\mathcal{T}, 5)$, $shift(\mathcal{T}, 6)$, $shift(\mathcal{T}, 7)$, and $shift(\mathcal{T}, 0)$.
Secondly, considering $H^{(3)}$, which is a combination of $H^{(1)}$ and $H^{(2)}$ and these two modifications are far enough. We can see the new candidates introduced by $H^{(3)}$ are either the objects from $shift(\mathcal{T}, 2)$ and $shift(\mathcal{T}, 3)$ by the $\lambda$-LCCS search of $H^{(1)}$ or those from $shift(\mathcal{T}, 5)$, $shift(\mathcal{T}$, $6)$, $shift(\mathcal{T}, 7)$, and $shift(\mathcal{T}, 0)$ by the $\lambda$-LCCS search of $H^{(2)}$. Since $H^{(1)}$ and $H^{(2)}$ have fewer modifications than $H^{(3)}$, they have higher priority than $H^{(3)}$. The new candidates introduced by $H^{(3)}$ were checked already. It is redundant to probe $H^{(3)}$.
\hfill $\triangle$ \par 
\end{example}

\begin{figure}[t]
\centering   
\includegraphics[width=0.47\textwidth]{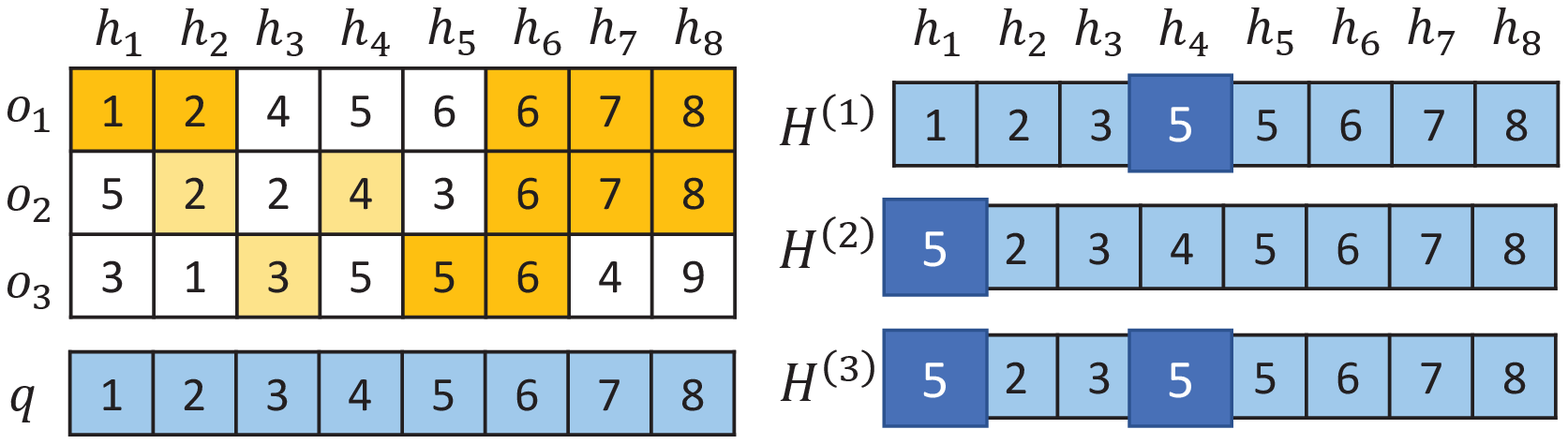}
\vspace{-1.0em}
\caption{An example of MP-LCCS-LSH}
\label{fig:mp-lccs-example}
\vspace{-0.5em}
\end{figure}

\vspace{-0.5em}
\paragraph{MP-LCCS-LSH}
To address these two problems, we design a multi-probe scheme for LCCS-LSH, named MP-LCCS-LSH, making use of existing multi-probe schemes, such as Multi-Probe LSH and FALCONN. Given a hash string $H(q)=[h_1(q)$, $h_2(q),\cdots,$ $h_m(q)]$, we can get $m$ lists of alternative hash values, i.e., $\{{h_1(q)}^{(j)}\},\{{h_2(q)}^{(j)}\},\cdots,\{{h_m(q)}^{(j)}\}$, where each $\{{h_i(q)}^{(j)}\}$ is a list of alternative hash values of $h_i(q)$. For example, for Multi-Probe LSH, $\{{h_i(q)}^{(j)}\} = \{h_i(q) \pm 1, h_i(q) \pm 2, \cdots\}$, whereas for FALCONN, $\{{h_i(q)}^{(j)}\}$ is a list of other vertices of the cross-polytope. Let $score(i, {h_i(q)}^{(j)})$ be the score of the $j^{th}$ alternative ${h_i(q)}^{(j)}$ in $i^{th}$ position, and we reuse the score function from existing multi-probe schemes. Without loss of generality, we consider each $\{{h_i(q)}^{(j)}\}$ is sorted in ascending order of their scores. A perturbation vector $\delta$ is a list of pairs $(i, {h_i(q)}^{(j)})$, where $i$ is the position of modification and ${h_i(q)}^{(j)}$ is used to replace $h_i(q)$, e.g., $\delta=\{(2, {h_2(q)}^{(1)})$, $(5, {h_5(q)}^{(3)})\}$ means to modify $h_2(q)$ to ${h_2(q)}^{(1)}$ and $h_5(q)$ to ${h_5(q)}^{(3)}$. We inherent $score(\delta)$ to be the score of $\delta$ from existing multi-probe schemes.

\subparagraph{Skip Unaffected Positions}
For the first problem, we skip the unaffected positions. 
During the first $\lambda$-LCCS search of $H(q)$, we additionally store the matched positions ${pos}_{l,i}, {pos}_{u,i}$ and the lengths ${len}_{l,i}$, ${len}_{u,i}$ for each $i$ at lines 2, 7, and 9 of Algorithm \ref{alg:query}. If the modification of $H(q)$ is not in the positions between $i$ and $i+\max({len}_{l,i}, {len}_{u,i})$, it will not affect the LCP of $shift(H(q), i-1)$ at position $i$. Thus, instead of conducting a full $\lambda$-LCCS search from position $1$ to $m$, the probe of $H(q)$ with $\delta=\{(i_1, h_{i_1}(q)^{(j_1)}),(i_2, h_{i_2}(q)^{(j_2)}),\cdots\}$ can be checked by the LCP of $shift(H(q), i)$ starting from the first position $i_s$ such that $i_s+\max({len}_{l,i_s}, {len}_{u,i_s}) > i_1$ to $i_1$, i.e., with a modification of Algorithm \ref{alg:query} at lines 2 and 5.
 
\begin{algorithm}[t]
\caption{Generating Perturbation Vectors}
\label{alg:multi-probe-lccs}
\KwIn{$\#probes$, $MAX\_GAP$;}
\KwOut{$\Delta$: a set of perturbation vectors.}
$\Delta \gets \{\emptyset\}$ \Comment*[r]{add "no perturbation" into $\Delta$}
$PQ \gets \emptyset$\Comment*[r]{a minimum priority queue}
\For {$i = 1$ to $m$} {
	$\delta = \{(i, {h_i(q)}^{(1)})\}$\;
	$PQ.push(\delta,score(\delta))$\;
}
\For {$t = 2$ to $\#probes$} {
	$s,\delta = PQ.top()$;  $PQ.pop()$\; 
	$\Delta \gets \Delta \cup \{\delta\}$\;
	$\delta_s = p\_shift(\delta)$\;
	$PQ.push(\delta_s, score(\delta_s))$\;
	\For {$gap = 1$ to $MAX\_GAP$} {
		$\delta_e = p\_expand(\delta, gap)$\;
		$PQ.push(\delta_e, score(\delta_e))$\;
	}
}
\Return $\Delta$\;
\end{algorithm}

\subparagraph{Perturbation Vector Generation}
For the second problem, we restrict the gap between two adjacent modified positions in a perturbation vector. For example, given a perturbation vector $\delta=\{(1,{h_1(q)}^{(3)}), (2,{h_2(q)}^{(1)}), (5,{h_5(q)}^{(2)})\}$, the gaps of $\delta$ at the $1^{st}$ and $2^{nd}$ positions are respectively $2-1=1$ and $5-2=3$, and they should be less than or equal to a threshold $MAX\_GAP$. We set $MAX\_GAP=2$ in practice.

Based on this heuristic idea, we propose a perturbation vector generation method in Algorithm \ref{alg:multi-probe-lccs}, within the similar shift-expand operations in \cite{lv2007multi}. We name them $p\_shift$ and $p\_expand$ to distinguish from the shift operation of CSA. Let $\delta=\{(i_1,h_{i_1}(q)^{(j_1)}), (i_2,h_{i_2}(q)^{(j_2)}), \cdots, (i_e,h_{i_e}(q)^{(j_e)})\}$. They are defined as follows:
\begin{itemize}
\item $p\_shift(\delta)$: use the next alternative hash value of the last modification operation of $\delta$, i.e., $p\_shift(\delta) = \{(i_1$, $h_{i_1}(q)^{(j_1)}),(i_2, h_{i_2}(q)^{(j_2)}),\cdots,(i_e, h_{i_e}(q)^{(j_e+1)})\}$;
\vspace{0.5em}
\item $p\_expand(\delta, gap)$: append $(i_e+gap, h_{i_e+gap}(q)^{(1)})$ to $\delta$, i.e., $p\_expand(\delta, gap)=\{(i_1, h_{i_1}(q)^{(j_1)}),(i_2, h_{i_2}(q)^{(j_2)}),$ $\cdots,(i_e, h_{i_e}(q)^{(j_e)}),(i_e+gap, h_{i_e+gap}(q)^{(1)})\}$.
\end{itemize}

\subparagraph{Remarks}
Even though the formulas of $p\_shift$ and $p\_expand$ are very similar to \cite{lv2007multi}, the meaning of the perturbation vector is different. Following the similar proofs from \cite{lv2007multi}, it can be shown that all perturbation vectors with gap less than $MAX\_GAP$ can be generated by Algorithm \ref{alg:multi-probe-lccs}, and they will be probed in ascending order of their scores.

\section{Theoretical Analysis}
\label{sect:analysis}

\begin{table*}[t]
\centering
\caption{Space and time complexities of E2LSH, C2LSH, and LCCS-LSH under different settings of $\alpha$}
\vspace{-0.75em}
\label{table:time_complexity}
\begin{tabular}{ccccccc} \toprule
Methods & $\alpha$ & $m$ & $\lambda$ & Space Complexity & Indexing Time Complexity & Query Time Complexity \\ \midrule
E2LSH \cite{datar2004locality} & -- & -- & -- & $O(n^{1+\rho})$ & $O(n^{1+\rho} \eta(d) \log n)$ & $O(n^\rho (\eta(d) \log n + d))$ \\
C2LSH \cite{gan2012locality} & -- & -- & -- & $O(n\log n)$ & $O(n\log n (\eta(d) + \log n))$ & $O(n\log n)$ \\
\multirow{3}{*}{LCCS-LSH}
& 0  & $O(1)$ & $O(n)$ & $O(n)$ & $O(n(\eta(d)+\log n))$ & $O(nd)$ \\ 
& 1  & $O(n^\rho)$ & $O(n^\rho)$ & $O(n^{1+\rho})$ & $O(n^{1+\rho} (\eta(d)+\log n))$ & $O(n^\rho (\eta(d)+d+\log n))$ \\
& $\tfrac{1}{1-\rho}$ & $O(n^{\frac{\rho}{1-\rho}})$ & $O(1)$ & $O(n^{\frac{1}{1-\rho}})$ & $O(n^{\frac{1}{1-\rho}} (\eta(d)+\log n))$ & $O(n^{\frac{\rho}{1-\rho}} (\eta(d)+\log n)+d)$ \\ \bottomrule
\end{tabular}
\end{table*}

\subsection{Quality Guarantee}
\label{sect:analysis:guarantee}
We now establish a quality guarantee for LCCS-LSH. Given any two strings $T=[t_1,t_2,\cdots,t_m]$ and $Q=[q_1,q_2,\cdots,q_m]$, suppose the probability that $t_i=q_i$ for each $i$ is independent and equals to $p$, i.e., $\Pr[t_i = q_i] = p$. Let $F_{m, p}(x)=\Pr[|LCCS(T,Q)| \leq x]$ be the CDF of the length of LCCS between $T$ and $Q$. Notably, $F_{m, p}(x)$ decreases monotonically as $p$ increases when $m$ and $x$ are fixed.

Let $B(q,R)$ and $\bar{B}(q,R)$ be the set of $\{o \in \mathcal{D} \mid Dist(o,q) \leq R\}$ and $\{o \in \mathcal{D} \mid Dist(o,q) > R\}$, respectively. Since $F_{m, p}(x)$ is monotonic w.r.t. $p$, we have Lemma \ref{lemma:amplification} as follows.
\begin{lemma}
\label{lemma:amplification}
Given a parameter $x$ such that $0 < x \leq m$, for any $o_i \in \bar{B}(q,cR)$, 
\begin{displaymath}
\Pr[|LCCS(H(o_i), H(q))| \leq x] \geq F_{m, p_2}(x), 
\end{displaymath}
and for any $o^*_j\in B(q, R)$, 
\begin{displaymath}
\Pr[|LCCS(H(o^*_j), H(q))| > x] \geq 1 - F_{m, p_1}(x).
\end{displaymath}
\end{lemma}

According to Lemma \ref{lemma:amplification}, if we want to demonstrate that LCCS-LSH enjoys the $(R,c)$-NNS with constant probability, we first need to study the property of $F_{m, p}(x)$. 

According to \cite{gordon1986extreme}, the longest consecutive heads in $n$ coin tosses with $\Pr[Head] = p$ can be asymptotically estimated by the largest value of $n(1-p)$ i.i.d. random variables that follow the exponential distribution. 
Thus, for a sufficiently large $m$, if we follow the similar constructions to LCCS-LSH except for the first random variable, $F_{m, p}(x)$ can also be modeled by the largest value of $m(1-p)$ i.i.d. random variables that follow the exponential distribution. Hence, 
\begin{lemma}
\label{lemma:prob_lccs}
Let $\hat{F}_{p}(x) = \exp(-p^{x})$ be the CDF of the extreme value distribution of $x$. As $m \rightarrow \infty$, $\forall x$, 
\begin{displaymath}
F_{m,p}(x) - \hat{F}_{p}(x-\log_{1/p}(m(1-p))) \rightarrow 0.
\end{displaymath}
\end{lemma}
\begin{proof}
Lemma \ref{lemma:prob_lccs} holds due to the Theorem 1 of \cite{gordon1986extreme}. It corresponds to the case when $k=0$ \cite{gordon1986extreme}.
\end{proof}

\begin{theorem}
\label{theorem:probability}
Given a distance metric $Dist(\cdot,\cdot)$ that admits an $(R,cR,p_1,p_2)$-sensitive LSH family $\mathcal{H}$, the LCCS-LSH scheme with hash length $m$ can answer the $(R,c)$-NNS over $Dist(\cdot,\cdot)$ by conducting $\lambda$-LCCS search with a probability at least $1/4$, where $\lambda=m^{1-1/\rho} n (1-p_1)^{-1/\rho} (1-p_2)(\ln 2)^{1/\rho}/p_2=O(m^{1-1/\rho}n)$ and $\rho=\ln(1/p_1)/\ln(1/p_2)$.
\end{theorem}

\begin{proof}
Let $\hat{F}_{m, p}(x)=\hat{F}_{p}(x-\log_{1/p}(m(1-p)))$. The median of $\hat{F}_{m, p}(x)$, $x_{1/2, p}$, can be computed as
\begin{equation}
\label{eqn:x-median}
x_{1/2, p} =\log_{p}(\ln (2) ) + \log_{1/p} m(1-p), 
\end{equation}
and the $(1-k/n)$ quantile of $\hat{F}_{m, p}(x)$ can be computed as
\begin{equation}
\label{eqn:x-k-quantile}
x_{1-k/n, p} = \log_{p}(-\ln (1-k/n) ) + \log_{1/p} m(1-p).
\end{equation}
	
Consider the case of verifying $k$ candidates from the $k$-LCCS search of $H(q)$. 
For a sufficiently large $n$, according to Lemma \ref{lemma:amplification} and the Central Limit Theory, the $k^{th}$ longest LCCS between $H(o_i)$ and $H(q)$ for $n$ objects $o_i\in \bar{B}(q, cR)$ is less than $x_{1-k/n, p_2}$ with a probability at least $1/2$. 
In addition, according to Lemma \ref{lemma:amplification}, any object $o^*\in B(q, R)$ has a longer LCCS than $x_{1/2, p_1}$ with a probability at least $1/2$. If the condition $x_{1/2, p_1}\geq x_{1-k/n, p_2}+1$ holds, there will be at least one object $o^*$ appeared in the $k$ candidates from $k$-LCCS search of $H(q)$ with a probability at least $1/4$. 
According to Equations \ref{eqn:x-median} and \ref{eqn:x-k-quantile}, when $n \rightarrow \infty$, the condition
\begin{displaymath}
\begin{split}
            & x_{1/2, p_1}\geq x_{1-k/n, p_2} +1\\
\impliedby  & -\ln(1-k/n) \geq  m^{1-1/\rho} (1-p_1)^{-1/\rho} (1-p_2) (\ln 2)^{1/\rho}/p_2\\
\impliedby  & k/n \geq  m^{1-1/\rho} (1-p_1)^{-1/\rho} (1-p_2) (\ln 2)^{1/\rho}/p_2\\
\impliedby  & k \geq  m^{1-1/\rho} n (1-p_1)^{-1/\rho} (1-p_2) (\ln 2)^{1/\rho}/p_2.
\end{split}
\end{displaymath}
Thus, by setting $\lambda = m^{1-1/\rho} n (1-p_1)^{-1/\rho} (1-p_2)(\ln 2)^{1/\rho}/p_2$, with a probability at least $1/4$: if $B(q,R) \neq \emptyset$, LCCS-LSH can get at least one $o \in B(q,cR)$ by conducting a $\lambda$-LCCS search; if $B(q,cR)=\emptyset$, LCCS-LSH can trivially return nothing as no candidate from $\lambda$-LCCS search is in $B(q,R)$. 
\end{proof}

\subsection{Space and Time Complexities}
\label{sect:analysis:complexity}
LCCS-LSH is LSH-family-independent and it can handle various kinds of distance metrics. Thus, we first discuss the complexities of distance computation and the computation of hash values. For simplicity, we assume the computation of $Dist(\cdot,\cdot)$ takes $O(d)$ time. The complexity of the computation of hash values for different LSH families is different. We assume computing each hash value takes $O(\eta(d))$ time. For example, the random projection LSH family \cite{datar2004locality} takes $O(d)$ time, the cross polytope LSH family \cite{terasawa2007spherical,andoni2015practical} requires $O(d \log d)$ time, whereas the random bits sampling LSH family \cite{indyk1998approximate} for Hamming distance only needs $\eta(d)=O(1)$.

According to Theorem \ref{theorem:klccs}, the query time complexity of Algorithm \ref{alg:query} is $O(\log n + (m + \lambda) \log m)$, and the space and indexing time complexities of Algorithm \ref{alg:indexing} are $O(mn)$ and $O(mn \log(n))$, respectively. In addition, in the indexing phase of LCCS-LSH, computing $nm$ hash values for $n$ data objects takes $O(nm \cdot \eta(d))$ time. In the query phase of LCCS-LSH, computing $m$ hash values for each query takes $O(m \cdot \eta(d))$ time and computing the actual distance for $\lambda$ candidates takes $O(\lambda d)$ time. According to Theorem \ref{theorem:probability}, $\lambda$ is determined by $m$ and $n$. By setting different $m$ values, LCCS-LSH has different time and space complexities. Thus, we introduce a parameter $\alpha$ to control the value of $m$ in different scales. According to Theorems \ref{theorem:klccs} and \ref{theorem:probability}, we have
\begin{corollary}
\label{corollary:lccs-lsh-complexity}
For any $0 \leq \alpha \leq \frac{1}{1-\rho}$, setting $m=O(n^{\alpha \rho})$, LCCS-LSH can answer the $(R,c)$-NNS with a probability at least $1/4$ using $O(n^{1+\alpha\rho})$ space, $O(n^{1+\alpha\rho} (\eta(d)+\log n))$ indexing time, and $O(n^{\alpha\rho}(\eta(d)+\alpha \log n) + n^{\alpha(\rho-1)+1}(d+\alpha \log n) )$ query time.
\vspace{-0.5em}
\end{corollary} 

The upper-bound of $\alpha$ is $\frac{1}{1-\rho}$, because $\lambda$ is at least 1. There are three typical settings of $\alpha$: (i) $\alpha=0$: the query time complexity of LCCS-LSH is equivalent to the complexity of linear scan; (ii) $\alpha=1$: compared with E2LSH \cite{datar2004locality} and C2LSH \cite{gan2012locality}, LCCS-LSH enjoys the least query time complexity; moreover, LCCS-LSH has the same space complexity as E2LSH and its index time complexity is also lower than that of E2LSH; C2LSH enjoys the least space and indexing time complexities, but its query time complexity is the largest among the three methods; (iii) $\alpha=\frac{1}{1-\rho}$: LCCS-LSH verifies only constant number of candidates, and thus it is suitable to the case that computing hash values is much cheaper than computing actual distances, e.g., the random bits sampling LSH family for Hamming distance in a very high dimensional space. Setting $\alpha$ in between $0$ and $\frac{1}{1-\rho}$ can smoothly control the trade-off between space and time complexities. We summarize the space and time complexities of E2LSH, C2LSH, and LCCS-LSH in Table \ref{table:time_complexity}.

As discussed in Section \ref{sect:preliminary:problem}, to get a data structure for $c$-ANNS, one should build multiple data structures for $(R,c)$-NNS with different $R \in \{1,c,c^2,\cdots\}$. This is because given an $(R,cR,p_1,p_2)$-sensitive LSH family $\mathcal{H}$, the parameters $K$ and $L$ in the static concatenating search framework depends on $p_1$ and $p_2$, which might be different when considering different $R$ values. 
For LCCS-LSH, given a fixed $\rho$, since $p_1$ and $p_2$ only affect $m$ by constant factors, it is possible to build \emph{one index} to handle variant $R$ values without changing the asymptotic time complexity. Specifically, given an $(R,cR,p_1,p_2)$-sensitive LSH family $\mathcal{H}$, if $\mathcal{H}$ satisfies the condition that $\rho_R \leq \rho^* < 1$ for all considered $R$ values, LCCS-LSH can handle $c$-ANNS using the same asymptotic time and space complexity as $(R,c)$-NNS. For example, the cross-polytope LSH family \cite{terasawa2007spherical,andoni2015practical} satisfies this condition, because it has the property that $\rho_R=\frac{1}{c^2} \frac{4-c^2 R^2}{4-R^2} +o(1) \leq \rho^* = \frac{1}{c^2} + o(1)$ for all $R$ values according to Corollary 1 of \cite{andoni2015practical}. On the other hand, the random projection LSH family \cite{datar2004locality} does not satisfy this condition since $\rho$ could be arbitrarily close to 1 for certain $R$ values once $w$ is fixed.

\section{Experiments}
\label{sect:expt}
In this section, we study the performance of LCCS-LSH and MP-LCCS-LSH over five real-life datasets for high dimensional $c$-ANNS. All methods are implemented in C++ and are compiled with g++ 8.3 using -O3 optimization. We conduct all experiments in a single thread on a machine with 8 Intel i7-3820 @ 3.60GHz CPUs and 64 GB RAM, running on Ubuntu 16.04.
\vspace{-0.25em}

\subsection{Datasets and Queries}
\label{sec:expt:setup:dataset}
We use five real-life datasets in our experiments, which cover a wide range of data types, including audio, image, text, and deep-learning data. We randomly select $100$ objects from their test sets and use them as queries. The statistics of datasets and queries are summarized in Table \ref{table:dataset}.
\begin{itemize}
\item \textbf{Msong.}\footnote{\url{http://www.ifs.tuwien.ac.at/mir/msd/download.html}.} The Msong dataset is a collection  of about $1$ million $420$-dimensional audio features and metadata for a contemporary popular music tracks.

\item \textbf{Sift.}\footnote{\url{http://corpus-texmex.irisa.fr/}.} The Sift dataset has $1$ million $128$-dimensional image sift features.

\item \textbf{Gist.}\footnote{\url{http://corpus-texmex.irisa.fr/}.} Gist is a $960$-dimensional dataset with $1$ million image gist features.

\item \textbf{GloVe.}\footnote{\url{https://nlp.stanford.edu/projects/glove/}.} It contains about $1.2$ million $100$-dimensional text embedding features extracted from Tweets.

\item \textbf{Deep.}\footnote{\url{https://github.com/DBWangGroupUNSW/nns_benchmark}.} It is a $256$-dimensional dataset that contains $1$ million deep neural codes of images obtained from the activations of a convolutional neural network.
\end{itemize}
\vspace{-0.25em}

\subsection{Evaluation Metrics}
\label{sec:expt:setup:metric}
We use the following metrics for performance evaluation. 
\begin{itemize}
\item \textbf{Index Size and Indexing Time.} We use the index size and indexing time to evaluate the indexing overhead of a method. The index size is defined by the memory usage for a method to build index. Similarly, the indexing time is defined as the wall-clock time for a method to build index.

\item \textbf{Recall.} We use recall to measure the accuracy of a method. For the $c$-$k$-ANNS, it is defined as the fraction of the total amount of data objects returned by a method that are appeared in the exact $k$ NNs.

\item \textbf{Ratio.} Overall ratio (or simply ratio) is also a popular measure to access the accuracy of a method. For the $c$-$k$-ANNS, it is defined as $\frac{1}{k} \sum_{i = 1}^k \tfrac{Dist(o_i, q)}{Dist(o_i^*, q)}$, where $o_i$ is the $i^{th}$ nearest object returned by a method and $o_i^*$ is the exact $i^{th}$ NN, where $i \in \{1,2,\cdots,k\}$. Intuitively, a smaller overall ratio means a higher accuracy.

\item \textbf{Query Time.} We consider the query time to evaluate the efficiency of a method. It is defined as the wall-clock time of a method to conduct a $c$-$k$-ANNS.
\end{itemize}

We report the average recall and ratio over all queries, and we run each method for each experiment five times to report its average running time and indexing overhead.

\begin{table}[tb]
\centering
\caption{Statistics of datasets and queries}
\vspace{-0.5em}
\label{table:dataset}
\begin{tabular}{cccccc} \toprule
Datasets & $\#$Objects & $\#$Queries & $d$ & Data Size & Type  \\ \midrule
Msong    & 992,272     & 100         & 420 & 1.6 GB    & Audio \\ 
Sift     & 1,000,000   & 100         & 128 & 488.3 MB  & Image \\
Gist     & 1,000,000   & 100         & 900 & 3.6 GB    & Image \\
GloVe    & 1,183,514   & 100         & 100 & 451.5 MB  & Text  \\ 
Deep     & 1,000,000   & 100         & 256 & 976.6 MB  & Deep  \\ \bottomrule
\end{tabular}
\end{table}

\begin{figure*}[tb]
\centering
\includegraphics[width=0.98\textwidth]{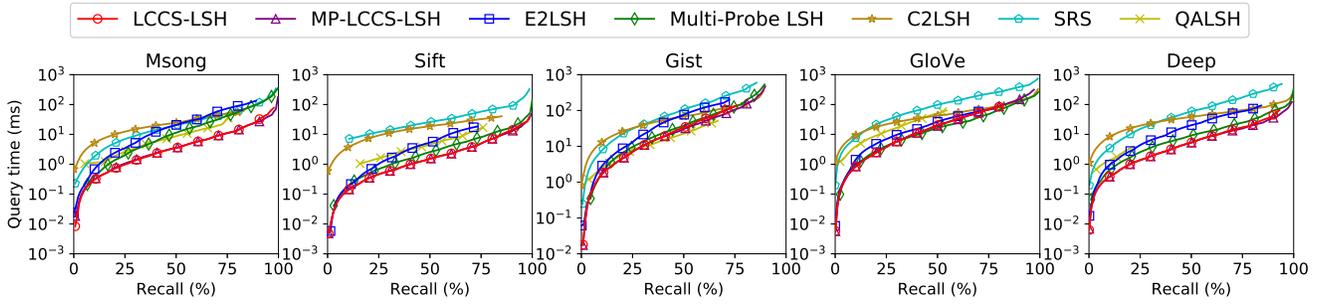}
\vspace{-0.5em}
\caption{Query time-recall curves (lower is better) of retrieving top-10 NNs under Euclidean distance.}
\label{fig:time-accuracy-euclidean}
\end{figure*}

\begin{figure*}[tb]
\centering
\includegraphics[width=0.98\textwidth]{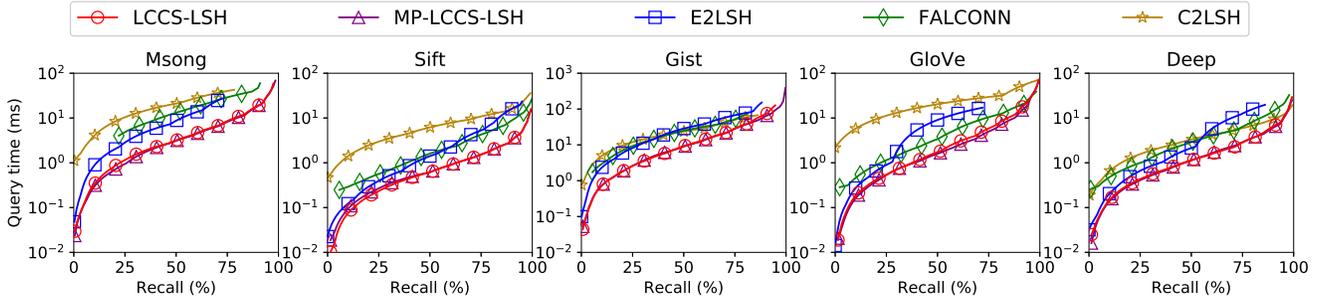}	
\vspace{-0.5em}
\caption{Query time-recall curves (lower is better) of retrieving top-10 NNs under Angular distance.}
\label{fig:time-accuracy-angular}
\vspace{-0.5em}
\end{figure*}

\subsection{Benchmark Methods}
\label{sec:expt:setup:benchmark}
Since LCCS-LSH is independent of LSH family and it supports $c$-ANNS with various kinds of distance metrics, we consider the random projection LSH family and cross-polytope LSH family and conduct experiments under two popular distance metrics, i.e., Euclidean distance and Angular distance. 

To make a fair comparison with different kinds of search framework, we select several state-of-the-art LSH schemes as benchmarks. Specifically, we evaluate the methods described as follows.
\begin{itemize}
\item \textbf{LCCS-LSH and MP-LCCS-LSH.} Both schemes adopt the LCCS search framework for $c$-$k$-ANNS. 
Compared to LCCS-LSH, we add an intelligent probing strategy to MP-LCCS-LSH to reduce the indexing overhead. 
We evaluate both schemes for $c$-$k$-ANNS under Euclidean distance and Angular distance, respectively.

\item \textbf{Multi-Probe LSH.} Multi-Probe LSH \cite{lv2007multi,dong2008modeling}  uses the static concatenating search framework with an intelligent probing strategy for $c$-$k$-ANNS. It is based on the random projection LSH family and is designed for Euclidean distance. We use a public implementation\footnote{\url{http://lshkit.sourceforge.net/}.} by the authors for performance evaluations.

\item \textbf{FALCONN.} Similar to Multi-Probe LSH, FALCONN \cite{andoni2015practical} also applies the static concatenating search framework with an intelligent probing strategy for $c$-$k$-ANNS. It is based on the cross polytope LSH family and is designed for Angular distance. We use a public implementation\footnote{\url{https://falconn-lib.org/}.} by the authors in the experiments.

\item \textbf{E2LSH.} E2LSH \cite{datar2004locality,andoni2005e2lsh} adopts the static concatenating search framework directly for $c$-$k$-ANNS. It is based on the random projection LSH family and is designed for Euclidean distance. To make a further comparison, we adapt it for Angular distance, where the LSH functions are drawn from the cross-polytope LSH family.

\item \textbf{C2LSH.} C2LSH \cite{gan2012locality} applies a dynamic collision counting framework for $c$-$k$-ANNS. Similar to E2LSH, it is designed for Euclidean distance. We also adapt it for Angular distance, where the LSH functions are drawn from the cross-polytope LSH family.

\item \textbf{SRS.} SRS \cite{sun2014srs} is a state-of-the-art LSH-based method which is designed for Euclidean distance. It converts data objects into low dimensions based on random projection and indexes them by a single R-tree for $c$-$k$-ANNS. We use its memory version\footnote{\url{https://github.com/DBWangGroupUNSW/SRS}.} with cover-tree for comparison.

\item \textbf{QALSH.} Similar to C2LSH, QALSH \cite{huang2017query} applies the dynamic collision counting framework for $c$-$k$-ANNS and it is designed for Euclidean distance. For the million-scale datasets, we use its memory version QALSH$^+$\footnote{\url{https://github.com/HuangQiang/QALSH_Mem}.} for comparison to reduce false positives.
\end{itemize}

For the $c$-$k$-ANNS, we set $k \in \{1,2,5,10,20,50,100\}$. 
To make a fair comparison, we fix the maximum number of LSH functions for all methods. Specifically, we set $K \in \{1,2$, $3,\cdots,10\}$ and $L \in \{8,16,32,\cdots,512\}$ for E2LSH, Multi-Probe LSH, and FALCONN such that $KL \leq 512$; we set $m \in \{8$, $16,32,\cdots,512\}$ and $l \in \{2,3,\cdots,10\}$ for C2LSH; we set the projected dimensions $d'\in \{4, 5, \cdots, 10\}$ for SRS; for QALSH, we adopt QALSH$^+$\footnote{We use kd-tree to split dataset into blocks and set $leaf=20,000$. We set $L \in \{10, 20, 30, 40\}$ projections and $L*M \in \{100, 240, 400\}$ boundary objects for each block as representative objects to determine close blocks.} and set $c\in \{2.0, 3.0, 4.0\}$ for each block to build QALSH. 
For LCCS-LSH and MP-LCCS-LSH, we set $m \in \{8$, $16,32,\cdots,512\}$ and $\#probes \in \{1,m+1,2m+1,4m+1,8m+1\}$. $w$ is fine-tuned for the random projection LSH family,\footnote{Specifically, $w$ is set to be 18.75, 226.0, 11294.0, 4.65, and 0.66 for the datasets Msong, Sift, Gist, GloVe, and Deep, respectively.} so that all methods achieve their best performance. 

\begin{figure*}[tb]
\centering
\includegraphics[width=1.0\textwidth]{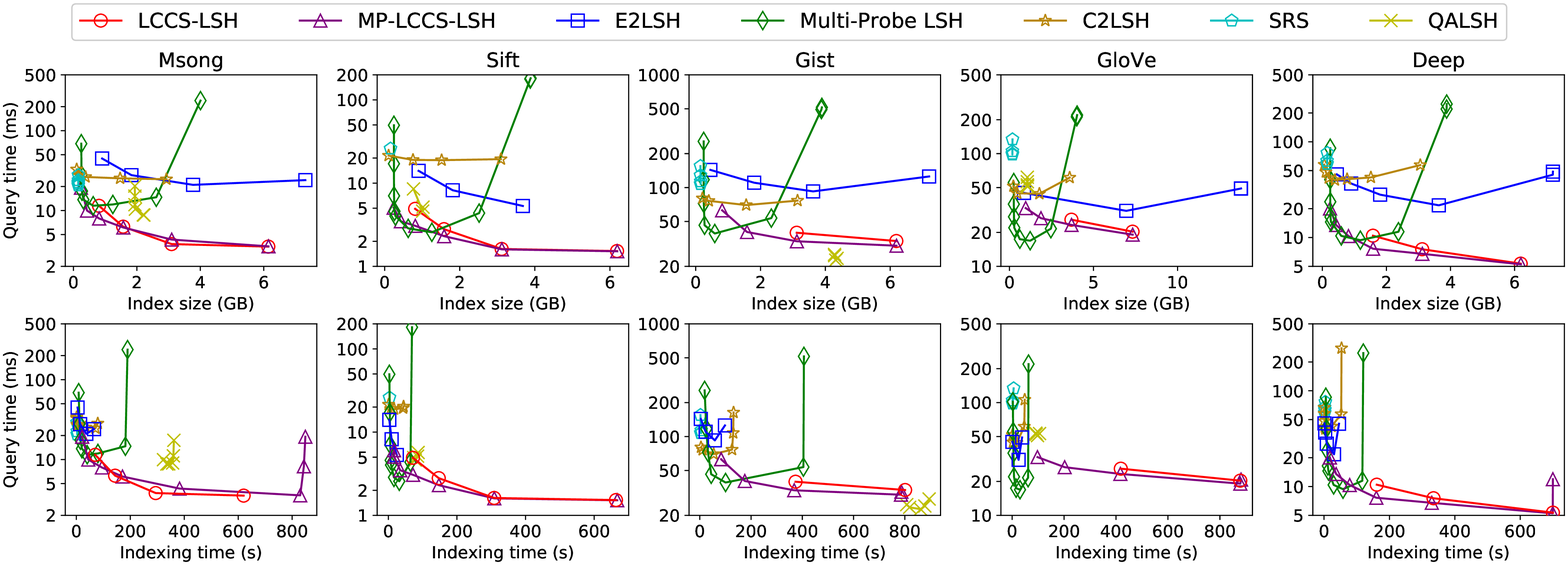}
\caption{Query time-Index size curves and Query time-Indexing time curves of retrieving top-10 NNs at $50\%$ recall level under Euclidean distance.}
\label{fig:time-index-euclidean}
\end{figure*}

\begin{figure*}[tb]
\centering
\includegraphics[width=1.0\textwidth]{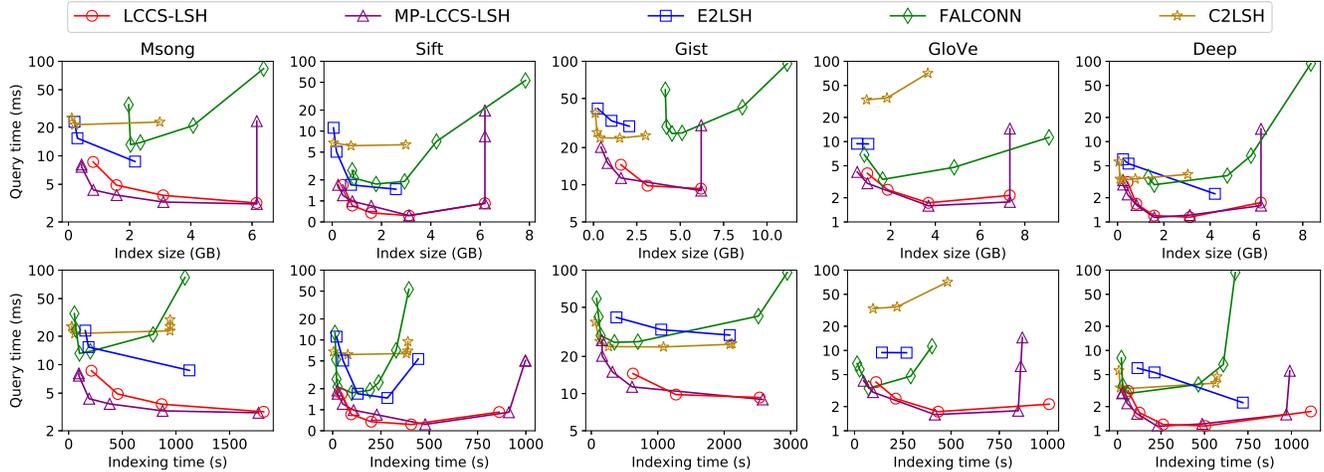}
\caption{Query time-Index size curves and Query time-Indexing time curves of retrieving top-10 NNs at $50\%$ recall level under Angular distance.}
\label{fig:time-index-angular}
\end{figure*}

\subsection{Results and Analysis}
\label{sec:expt:analysis}
We study the performance of LCCS-LSH and MP-LCCS-LSH in terms of five aspects: the query performance, indexing performance, the sensitivity to $k$, the impact of $m$, and the impact of $\#probes$. 

\paragraph{Query Performance}
We first study the query performance of LCCS-LSH and MP-LCCS-LSH. Different number of candidates and probes are used for all methods to achieve different recall levels. To remove the impact of parameters for each method, we report their lowest query time for all combinations of parameters under each certain recall level using grid search. We consider $k = 10$. The query time-recall curves under Euclidean distance and Augular distance are shown in Figures \ref{fig:time-accuracy-euclidean} and \ref{fig:time-accuracy-angular}, respectively. Similar trends can be observed from other $k$ values.

From Figure \ref{fig:time-accuracy-euclidean}, we observe LCCS-LSH and MP-LCCS-LSH achieve the best or nearly the best performance under Euclidean distance. 
Compared with E2LSH, QALSH, and Multi-Probe LSH, even though they are close to each other over Gist and GloVe, LCCS-LSH and MP-LCCS-LSH achieve around $240\%$ acceleration over Msong, $70\%$ acceleration over Sift, and $80\%$ acceleration over Deep under certain recall level. These results also demonstrate the efficiency of CSA for the LCCS search framework in the sense that identifying objects with the maximum length of LCCS is as efficient as hash table lookups. Furthermore, Multi-Probe LSH enjoys a slightly better trade-off between efficiency and accuracy than E2LSH, which satisfies the observations from \cite{datar2004locality, andoni2005e2lsh}. 
Compared with C2LSH and SRS, both LCCS-LSH and MP-LCCS-LSH achieve at least one order of magnitude acceleration under certain recall level for all of the five datasets, because the query time complexity of C2LSH is much worse than that of LCCS-LSH and SRS makes use of tree-based method to retrieve the candidates which is not as efficient as CSA. The performance of LCCS-LSH and MP-LCCS-LSH are close to each other. This is because although MP-LCCS-LSH introduces more probing overhead, it checks for fewer candidates than LCCS-LSH at the same recall level due to its intelligent probing strategy.

From Figure \ref{fig:time-accuracy-angular}, similar to the results under Euclidean distance, the performance of LCCS-LSH and MP-LCCS-LSH under Augular distance is better than those of other methods among all datasets, and their advantages are more apparent. Specifically, LCCS-LSH and MP-LCCS-LSH achieve at least $100\%$ acceleration compared with the second fastest competitor under $50\%$ recall level for all datasets. 
Furthermore, the performance of FALCONN is sightly better than that of E2LSH, especially when the recall level is high, which also fits the observations from \cite{andoni2015practical}. The query time-ratio curves show similar trends to the query time-recall curves. To be concise, we omit those results here.

\begin{figure}[tb]
\centering
\includegraphics[width=0.47\textwidth]{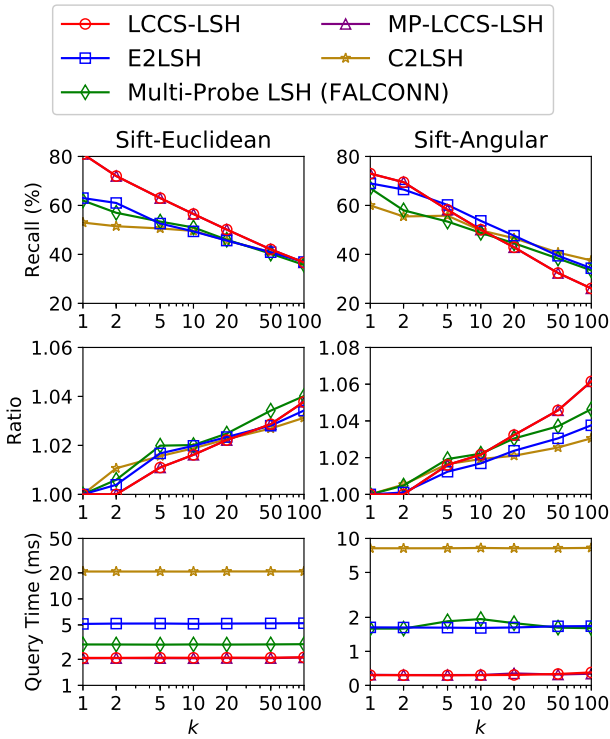}
\vspace{-1.0em}
\caption{Query performance vs.\ $k$.}
\label{fig:query-vs-k}
\vspace{-0.5em}
\end{figure}

\paragraph{Indexing Performance}
We then study the indexing performance of LCCS-LSH and MP-LCCS-LSH. We continue to consider $k = 10$. Since different parameters are used at different recall levels, we present the lowest query time under different index size (or indexing time) of all methods at $50\%$ recall level, to show the trade-off between query time and index size (or indexing time). The results under Euclidean distance and Angular distance are displayed in Figures \ref{fig:time-index-euclidean} and \ref{fig:time-index-angular}, respectively.\footnote{We do not show the results if the methods do not achieve $50\%$ recall level in Figures \ref{fig:time-accuracy-euclidean} and \ref{fig:time-accuracy-angular}.} Similar trends can be observed from other recall levels.

\begin{figure}[t]
\centering
\includegraphics[width=0.47\textwidth]{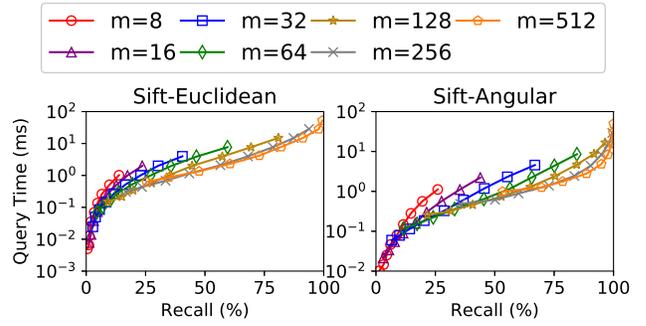}
\vspace{-1.0em}
\caption{Impact of\ $m$ for LCCS-LSH.}
\label{fig:query-vs-m}
\end{figure}

\begin{figure}[t]
\centering
\includegraphics[width=0.47\textwidth]{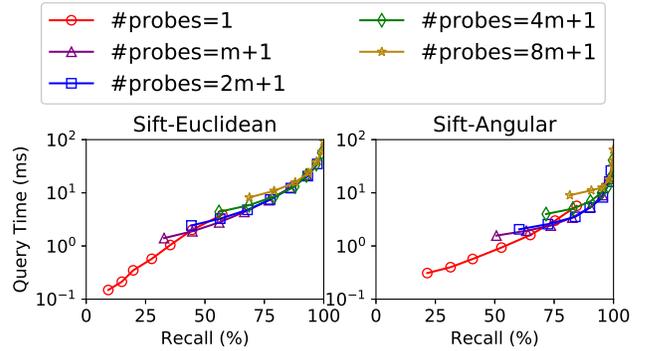}
\vspace{-1.0em}
\caption{Impact of\ $\#probes$ for MP-LCCS-LSH.}
\label{fig:query-vs-p}
\vspace{-1.0em}
\end{figure}

Figure \ref{fig:time-index-euclidean} shows that MP-LCCS-LSH enjoys better trade-off between query time and indexing overhead than LCCS-LSH, especially when only few memory is used. This means that the intelligent probing strategy of MP-LCCS-LSH can help to produce more candidates efficiently when $m$ is relatively small. 
Among all of the seven methods, Multi-Probe LSH is competitive in terms of the trade-off between query time and index size as it is designed to save space without losing too much information. For Gist and GloVe, Multi-Probe LSH uses less indexing overhead than MP-LCCS-LSH under the same query time, whereas for Msong, Sift, and Deep, MP-LCCS-LSH takes less query time under the same indexing budget if more memory is allowed. Compared to other competitors, MP-LCCS-LSH enjoys a better trade-off between query time and indexing overhead. The reasons are as follows: due to the static concatenating search framework, E2LSH cannot share LSH functions between different hash tables, which leads to a large indexing overhead; C2LSH, SRS, and QALSH can achieve good performance with small indexing budget, but they cannot significantly reduce the query time when more memory is allowed. 

We also observe that increasing index size (and indexing time) cannot always reduce the query time of all methods, because certain index size is good enough to achieve $50\%$ recall. In this case, using more hash functions will introduce more memory and more indexing time. Similar pattern can be observed from Figure \ref{fig:time-index-angular} under Angular distance. 

\paragraph{Sensitivity to $\bm{k}$}
Next, we study sensitivity to $k$ for the query performance of LCCS-LSH and MP-LCCS-LSH in terms of recall, ratio, and query time. We consider $k \in \{1,2,5,10,20$, $50$, $100\}$. To remove the impact of parameters for each method, we present their best query performance vs. $k$ for all combinations of parameters under the similar recall levels. The results of all methods over Sift under Euclidean distance and Angular distance are shown in Figure \ref{fig:query-vs-k}. Similar trends can be observed from other datasets. 

From Figure \ref{fig:query-vs-k}, except for C2LSH, the slope of LCCS-LSH and MP-LCCS-LSH on $k$ is similar to those of other competitors. Thus, LCCS-LSH and MP-LCCS-LSH are at least as stable as other methods. For C2LSH, it is more stable than others, because it requires $O(n)$ query time to carefully select candidates and it is inefficient. Furthermore, under the similar recall levels, the ratios of all methods are close to each other, whereas LCCS-LSH and MP-LCCS-LSH enjoys less query time than other competitors, which is consistent with the results presented in Figures \ref{fig:time-accuracy-euclidean} and \ref{fig:time-accuracy-angular}.

\paragraph{Impact of $\bm{m}$}
We study the impact of $m$ for LCCS-LSH. Figure \ref{fig:query-vs-m} shows the query time at different recall levels of LCCS-LSH by setting different $m \in \{8, 16, 32, 64, 128, 256, 512\}$ for Sift dataset under Euclidean distance and Angular distance. Similar trends can be observed for other datasets. 

As can be seen from Figure \ref{fig:query-vs-m}, LCCS-LSH achieves different trade-off between query time and recall under different settings of $m$, and in general, a larger $m$ lead to less query time at the same recall levels especially when the recalls are high. Furthermore, at certain recall level, e.g., $25\%$ for Sift under Euclidean distance, increasing $m$ will no longer decrease the query time. It means that at this recall level, the corresponding $m$ is optimal among all considered $m$'s for LCCS-LSH, e.g., $m=256$ is optimal for Sift at $25\%$ recall.

\paragraph{Impact of $\bm{\#probes}$}
Finally, we study the impact of $\#probes$ for MP-LCCS-LSH. We set $m=128$ and consider $\#probes$ ranging from $\{1, m+1, 2m+1, 4m+1, 8m+1\}$.\footnote{MP-LCCS-LSH is equivalent to LCCS-LSH when $\#probes=1$.} Figure \ref{fig:query-vs-p} shows the query time of MP-LCCS-LSH over Sift at different recall levels for different $\#probes$. Similar trends can be observed from other $m$'s and other datasets.

From Figure \ref{fig:query-vs-p}, we observe that MP-LCCS-LSH can accelerate the $c$-$k$-ANNS of LCCS-LSH at relatively high recall levels, where LCCS-LSH needs to check more candidates than MP-LCCS-LSH. However, for the lower recall levels, since the cost of each probe is higher than the time spent on verification, LCCS-LSH is better than MP-LCCS-LSH. This also confirms the results from Figures \ref{fig:time-index-euclidean} and \ref{fig:time-index-angular} that MP-LCCS-LSH can reduce indexing overhead of LCCS-LSH but can hardly improve the query time when LCCS-LSH uses sufficient memory.
\vspace{-0.5em}

\subsection{Summary}
\label{sec:expt:summary} 
Based on the experimental results, we have three important observations. 
Firstly, both LCCS-LSH and MP-LCCS-LSH are able to answer $c$-$k$-ANNS under Euclidean distance and Angular distance, which verifies their flexibilities to support various kinds of distance metrics. 
Secondly, both LCCS-LSH and MP-LCCS-LSH outperforms state-of-the-art methods, such as Multi-Probe LSH, FALCONN, E2LSH, and C2LSH. Specifically, LCCS-LSH and MP-LCCS-LSH enjoy a better trade-off between efficiency and accuracy than other competitors. In addition, in most of datasets, they also have a better trade-off between the query time and indexing overhead than Multi-Probe LSH, FALCONN, E2LSH and C2LSH. Finally, MP-LCCS-LSH is better than LCCS-LSH. Both schemes have almost the same trade-off between efficiency and accuracy, but MP-LCCS-LSH enjoys a better trade-off between query time and indexing overhead than LCCS-LSH.

\section{Related Work}
\label{sect:related_work}
NNS is a classic problem and is ubiquitous in various fields. 
The exact NNS in low-dimensional space is well solved by the tree-based methods \cite{guttman1984rtree,bentley1990k,katayama1997sr}. Due to the ``curse of dimensionality," these solutions cannot scale up to high dimensional space. 
Since the schemes we proposed are LSH-based methods, we focus on LSH schemes for high-dimensional $c$-ANNS.

LSH was originally introduced by Indyk and Motwani \cite{indyk1998approximate} for Hamming space, and later was extended to other distance metrics such as Jaccard similarity \cite{broder1997resemblance}, Angular distance \cite{charikar2002similarity,terasawa2007spherical}, and $l_p$ distance \cite{datar2004locality}. Although extensive studies have been done for the LSH families, the search framework behind the LSH families is less investigated. Existing works for the search framework can be roughly divided into two categories: static concatenating search framework and dynamic collision counting framework and their variants.

The first category is the static concatenating search framework and its variants. This search framework is most widely used in the LSH literatures. There are two potential problems behind this search framework. Firstly, the setting of $K$ is sensitive to $R$, and hence it needs to be tuned for every dataset. Secondly, the theoretical $L$ is usually prohibitively large. Thus, many variants of this search framework have been proposed to address these two issues.

To make $K$ suitable for different $R$ values and datasets, LSH-Forest \cite{bawa2005lsh} concatenates hash values into a sequence instead of a single hash value, so that the LCP between the hash values of query and data objects can be found via a trie structure. LSB-Forest \cite{tao2009quality} uses a z-order curve to encode the hash values and uses a B-tree to index the hash codes. Hence, the $K$ value can be conceptually automatically decided for different $R$. Similarly, SK-LSH \cite{liu2014sk} sorts the compound keys in alphabetical order, and thus it can reduce the I/O costs for external storages. Compared with these methods, LCCS-LSH uses the data structure CSA to store hash values. Since CSA can reuse the hash values in every position, it carries more information than sequence and curves. From this perspective, LCCS-LSH can be considered to extend them by virtually building more trees. 
To reduce the large $L$, Multi-Probe LSH \cite{lv2007multi} is proposed to heuristically boost conceptual $L$ by probing more buckets without extra memory usage. FALCONN \cite{andoni2015practical} further demonstrates its effectiveness on another LSH family. We also propose a new Multi-Probe scheme MP-LCCS-LSH to boost the conceptual $L$ and reduce the indexing overhead.

Another category is the dynamic collision counting framework. C2LSH \cite{gan2012locality} uses the number of identical hash values that data objects and query collide as the indicator of their actual distance. QALSH \cite{huang2015query,huang2017query} further extends this idea by considering real number as ``hash value." The counting-based indicator, although can identify the near neighbors of query precisely, unavoidably requires to count a large number of false positives, which limits its scalability when $n$ is very large. In contrast, LCCS-LSH can be understood as a dynamic concatenating search framework. By leveraging CSA for $k$-LCCS search, LCCS-LSH is able to answer $c$-ANNS with \emph{sublinear} query time and sub-quadratic space.

\section{Conclusion}
\label{sect:conclusion}
In this paper, we introduce a novel LSH scheme LCCS-LSH for high-dimensional $c$-ANNS with a theoretical guarantee. 
We define a new concept of LCCS and propose a novel data structure CSA for $k$-LCCS search. CSA is potentially of separate interest for other fields of computer science. 
LCCS-LSH adopts the LCCS search framework to dynamically concatenate consecutive hash values, which yields a simple yet effective way to identify the close objects. 
It requires to tune a single parameter $m$ only, which is unavoidable for the trade-off between space and query time.  
In addition, we propose MP-LCCS-LSH to further reduce the indexing overhead. Extensive experiments over five real-life datasets demonstrate the superior performance of LCCS-LSH and MP-LCCS-LSH. 

\begin{acks}
This research is supported by the National Research Foundation, Singapore under its Strategic Capability Research Centres Funding Initiative and the National Research Foundation Singapore under its AI Singapore Programme. Any opinions, findings and conclusions or recommendations expressed in this material are those of the author(s) and do not reflect the views of National Research Foundation, Singapore.
\end{acks}

\balance
\bibliographystyle{ACM-Reference-Format}
\bibliography{FullPaper}


\begin{thebibliography}{40}


\ifx \showCODEN    \undefined \def \showCODEN     #1{\unskip}     \fi
\ifx \showDOI      \undefined \def \showDOI       #1{#1}\fi
\ifx \showISBNx    \undefined \def \showISBNx     #1{\unskip}     \fi
\ifx \showISBNxiii \undefined \def \showISBNxiii  #1{\unskip}     \fi
\ifx \showISSN     \undefined \def \showISSN      #1{\unskip}     \fi
\ifx \showLCCN     \undefined \def \showLCCN      #1{\unskip}     \fi
\ifx \shownote     \undefined \def \shownote      #1{#1}          \fi
\ifx \showarticletitle \undefined \def \showarticletitle #1{#1}   \fi
\ifx \showURL      \undefined \def \showURL       {\relax}        \fi
\providecommand\bibfield[2]{#2}
\providecommand\bibinfo[2]{#2}
\providecommand\natexlab[1]{#1}
\providecommand\showeprint[2][]{arXiv:#2}

\bibitem[\protect\citeauthoryear{Andoni}{Andoni}{2005}]%
        {andoni2005e2lsh}
\bibfield{author}{\bibinfo{person}{Alexandr Andoni}.}
  \bibinfo{year}{2005}\natexlab{}.
\newblock \showarticletitle{E2LSH 0.1 User manual}.
\newblock
  \bibinfo{journal}{\emph{http://web.mit.edu/andoni/www/LSH/index.html}}
  (\bibinfo{year}{2005}).
\newblock


\bibitem[\protect\citeauthoryear{Andoni and Indyk}{Andoni and Indyk}{2006}]%
        {andoni2006near}
\bibfield{author}{\bibinfo{person}{Alexandr Andoni} {and}
  \bibinfo{person}{Piotr Indyk}.} \bibinfo{year}{2006}\natexlab{}.
\newblock \showarticletitle{Near-optimal hashing algorithms for approximate
  nearest neighbor in high dimensions}. In \bibinfo{booktitle}{\emph{FOCS}}.
  \bibinfo{pages}{459--468}.
\newblock


\bibitem[\protect\citeauthoryear{Andoni, Indyk, Laarhoven, Razenshteyn, and
  Schmidt}{Andoni et~al\mbox{.}}{2015}]%
        {andoni2015practical}
\bibfield{author}{\bibinfo{person}{Alexandr Andoni}, \bibinfo{person}{Piotr
  Indyk}, \bibinfo{person}{Thijs Laarhoven}, \bibinfo{person}{Ilya
  Razenshteyn}, {and} \bibinfo{person}{Ludwig Schmidt}.}
  \bibinfo{year}{2015}\natexlab{}.
\newblock \showarticletitle{Practical and optimal LSH for angular distance}. In
  \bibinfo{booktitle}{\emph{NeurIPS}}. \bibinfo{pages}{1225--1233}.
\newblock


\bibitem[\protect\citeauthoryear{Andoni and Razenshteyn}{Andoni and
  Razenshteyn}{2015}]%
        {andoni2015optimal}
\bibfield{author}{\bibinfo{person}{Alexandr Andoni} {and} \bibinfo{person}{Ilya
  Razenshteyn}.} \bibinfo{year}{2015}\natexlab{}.
\newblock \showarticletitle{Optimal data-dependent hashing for approximate near
  neighbors}. In \bibinfo{booktitle}{\emph{STOC}}. \bibinfo{pages}{793--801}.
\newblock


\bibitem[\protect\citeauthoryear{Bawa, Condie, and Ganesan}{Bawa
  et~al\mbox{.}}{2005}]%
        {bawa2005lsh}
\bibfield{author}{\bibinfo{person}{Mayank Bawa}, \bibinfo{person}{Tyson
  Condie}, {and} \bibinfo{person}{Prasanna Ganesan}.}
  \bibinfo{year}{2005}\natexlab{}.
\newblock \showarticletitle{LSH forest: self-tuning indexes for similarity
  search}. In \bibinfo{booktitle}{\emph{WWW}}. \bibinfo{pages}{651--660}.
\newblock


\bibitem[\protect\citeauthoryear{Bentley}{Bentley}{1990}]%
        {bentley1990k}
\bibfield{author}{\bibinfo{person}{Jon~Louis Bentley}.}
  \bibinfo{year}{1990}\natexlab{}.
\newblock \showarticletitle{K-d trees for semidynamic point sets}. In
  \bibinfo{booktitle}{\emph{SoCG}}. \bibinfo{pages}{187--197}.
\newblock


\bibitem[\protect\citeauthoryear{Beygelzimer, Kakade, and Langford}{Beygelzimer
  et~al\mbox{.}}{2006}]%
        {beygelzimer2006cover}
\bibfield{author}{\bibinfo{person}{Alina Beygelzimer}, \bibinfo{person}{Sham
  Kakade}, {and} \bibinfo{person}{John Langford}.}
  \bibinfo{year}{2006}\natexlab{}.
\newblock \showarticletitle{Cover trees for nearest neighbor}. In
  \bibinfo{booktitle}{\emph{ICML}}. \bibinfo{pages}{97--104}.
\newblock


\bibitem[\protect\citeauthoryear{Broder}{Broder}{1997}]%
        {broder1997resemblance}
\bibfield{author}{\bibinfo{person}{Andrei~Z Broder}.}
  \bibinfo{year}{1997}\natexlab{}.
\newblock \showarticletitle{On the resemblance and containment of documents}.
  In \bibinfo{booktitle}{\emph{Proceedings of Compression and Complexity of
  Sequences}}. \bibinfo{pages}{21--29}.
\newblock


\bibitem[\protect\citeauthoryear{Broder, Charikar, Frieze, and
  Mitzenmacher}{Broder et~al\mbox{.}}{1998}]%
        {broder1998min}
\bibfield{author}{\bibinfo{person}{Andrei~Z Broder}, \bibinfo{person}{Moses
  Charikar}, \bibinfo{person}{Alan~M Frieze}, {and} \bibinfo{person}{Michael
  Mitzenmacher}.} \bibinfo{year}{1998}\natexlab{}.
\newblock \showarticletitle{Min-wise independent permutations}. In
  \bibinfo{booktitle}{\emph{STOC}}. \bibinfo{pages}{327--336}.
\newblock


\bibitem[\protect\citeauthoryear{Charikar}{Charikar}{2002}]%
        {charikar2002similarity}
\bibfield{author}{\bibinfo{person}{Moses~S Charikar}.}
  \bibinfo{year}{2002}\natexlab{}.
\newblock \showarticletitle{Similarity estimation techniques from rounding
  algorithms}. In \bibinfo{booktitle}{\emph{STOC}}. \bibinfo{pages}{380--388}.
\newblock


\bibitem[\protect\citeauthoryear{Datar, Immorlica, Indyk, and Mirrokni}{Datar
  et~al\mbox{.}}{2004}]%
        {datar2004locality}
\bibfield{author}{\bibinfo{person}{Mayur Datar}, \bibinfo{person}{Nicole
  Immorlica}, \bibinfo{person}{Piotr Indyk}, {and} \bibinfo{person}{Vahab~S
  Mirrokni}.} \bibinfo{year}{2004}\natexlab{}.
\newblock \showarticletitle{Locality-sensitive hashing scheme based on p-stable
  distributions}. In \bibinfo{booktitle}{\emph{SoCG}}.
  \bibinfo{pages}{253--262}.
\newblock


\bibitem[\protect\citeauthoryear{Dong, Wang, Josephson, Charikar, and Li}{Dong
  et~al\mbox{.}}{2008}]%
        {dong2008modeling}
\bibfield{author}{\bibinfo{person}{Wei Dong}, \bibinfo{person}{Zhe Wang},
  \bibinfo{person}{William Josephson}, \bibinfo{person}{Moses Charikar}, {and}
  \bibinfo{person}{Kai Li}.} \bibinfo{year}{2008}\natexlab{}.
\newblock \showarticletitle{Modeling LSH for performance tuning}. In
  \bibinfo{booktitle}{\emph{CIKM}}. \bibinfo{pages}{669--678}.
\newblock


\bibitem[\protect\citeauthoryear{Fagin, Kumar, and Sivakumar}{Fagin
  et~al\mbox{.}}{2003}]%
        {fagin2003efficient}
\bibfield{author}{\bibinfo{person}{Ronald Fagin}, \bibinfo{person}{Ravi Kumar},
  {and} \bibinfo{person}{Dandapani Sivakumar}.}
  \bibinfo{year}{2003}\natexlab{}.
\newblock \showarticletitle{Efficient similarity search and classification via
  rank aggregation}. In \bibinfo{booktitle}{\emph{SIGMOD}}.
  \bibinfo{pages}{301--312}.
\newblock


\bibitem[\protect\citeauthoryear{Fu, Xiang, Wang, and Cai}{Fu
  et~al\mbox{.}}{2019}]%
        {fu2019fast}
\bibfield{author}{\bibinfo{person}{Cong Fu}, \bibinfo{person}{Chao Xiang},
  \bibinfo{person}{Changxu Wang}, {and} \bibinfo{person}{Deng Cai}.}
  \bibinfo{year}{2019}\natexlab{}.
\newblock \showarticletitle{Fast approximate nearest neighbor search with the
  navigating spreading-out graph}.
\newblock \bibinfo{journal}{\emph{PVLDB}} \bibinfo{volume}{12},
  \bibinfo{number}{5} (\bibinfo{year}{2019}), \bibinfo{pages}{461--474}.
\newblock


\bibitem[\protect\citeauthoryear{Gan, Feng, Fang, and Ng}{Gan
  et~al\mbox{.}}{2012}]%
        {gan2012locality}
\bibfield{author}{\bibinfo{person}{Junhao Gan}, \bibinfo{person}{Jianlin Feng},
  \bibinfo{person}{Qiong Fang}, {and} \bibinfo{person}{Wilfred Ng}.}
  \bibinfo{year}{2012}\natexlab{}.
\newblock \showarticletitle{Locality-sensitive hashing scheme based on dynamic
  collision counting}. In \bibinfo{booktitle}{\emph{SIGMOD}}.
  \bibinfo{pages}{541--552}.
\newblock


\bibitem[\protect\citeauthoryear{Gionis, Indyk, Motwani, et~al\mbox{.}}{Gionis
  et~al\mbox{.}}{1999}]%
        {gionis1999similarity}
\bibfield{author}{\bibinfo{person}{Aristides Gionis}, \bibinfo{person}{Piotr
  Indyk}, \bibinfo{person}{Rajeev Motwani}, {et~al\mbox{.}}}
  \bibinfo{year}{1999}\natexlab{}.
\newblock \showarticletitle{Similarity search in high dimensions via hashing}.
  In \bibinfo{booktitle}{\emph{VLDB}}, Vol.~\bibinfo{volume}{99}.
  \bibinfo{pages}{518--529}.
\newblock


\bibitem[\protect\citeauthoryear{Gordon, Schilling, and Waterman}{Gordon
  et~al\mbox{.}}{1986}]%
        {gordon1986extreme}
\bibfield{author}{\bibinfo{person}{Louis Gordon}, \bibinfo{person}{Mark~F
  Schilling}, {and} \bibinfo{person}{Michael~S Waterman}.}
  \bibinfo{year}{1986}\natexlab{}.
\newblock \showarticletitle{An extreme value theory for long head runs}.
\newblock \bibinfo{journal}{\emph{Probability Theory and Related Fields}}
  \bibinfo{volume}{72}, \bibinfo{number}{2} (\bibinfo{year}{1986}),
  \bibinfo{pages}{279--287}.
\newblock


\bibitem[\protect\citeauthoryear{Guttman}{Guttman}{1984}]%
        {guttman1984rtree}
\bibfield{author}{\bibinfo{person}{Antonin Guttman}.}
  \bibinfo{year}{1984}\natexlab{}.
\newblock \showarticletitle{R-trees: A dynamic index structure for spatial
  searching}. In \bibinfo{booktitle}{\emph{SIGMOD}}. \bibinfo{pages}{47--57}.
\newblock


\bibitem[\protect\citeauthoryear{Har-Peled, Indyk, and Motwani}{Har-Peled
  et~al\mbox{.}}{2012}]%
        {har2012approximate}
\bibfield{author}{\bibinfo{person}{Sariel Har-Peled}, \bibinfo{person}{Piotr
  Indyk}, {and} \bibinfo{person}{Rajeev Motwani}.}
  \bibinfo{year}{2012}\natexlab{}.
\newblock \showarticletitle{Approximate nearest neighbor: Towards removing the
  curse of dimensionality}.
\newblock \bibinfo{journal}{\emph{Theory of Computing}} \bibinfo{volume}{8},
  \bibinfo{number}{1} (\bibinfo{year}{2012}), \bibinfo{pages}{321--350}.
\newblock


\bibitem[\protect\citeauthoryear{Hinneburg, Aggarwal, and Keim}{Hinneburg
  et~al\mbox{.}}{2000}]%
        {hinneburg2000nearest}
\bibfield{author}{\bibinfo{person}{Alexander Hinneburg},
  \bibinfo{person}{Charu~C Aggarwal}, {and} \bibinfo{person}{Daniel~A Keim}.}
  \bibinfo{year}{2000}\natexlab{}.
\newblock \showarticletitle{What is the nearest neighbor in high dimensional
  spaces?}. In \bibinfo{booktitle}{\emph{VLDB}}. \bibinfo{pages}{506--515}.
\newblock


\bibitem[\protect\citeauthoryear{Huang, Feng, Fang, Ng, and Wang}{Huang
  et~al\mbox{.}}{2017}]%
        {huang2017query}
\bibfield{author}{\bibinfo{person}{Qiang Huang}, \bibinfo{person}{Jianlin
  Feng}, \bibinfo{person}{Qiong Fang}, \bibinfo{person}{Wilfred Ng}, {and}
  \bibinfo{person}{Wei Wang}.} \bibinfo{year}{2017}\natexlab{}.
\newblock \showarticletitle{Query-aware locality-sensitive hashing scheme for
  $l_p$ norm}.
\newblock \bibinfo{journal}{\emph{VLDBJ}} \bibinfo{volume}{26},
  \bibinfo{number}{5} (\bibinfo{year}{2017}), \bibinfo{pages}{683--708}.
\newblock


\bibitem[\protect\citeauthoryear{Huang, Feng, Zhang, Fang, and Ng}{Huang
  et~al\mbox{.}}{2015}]%
        {huang2015query}
\bibfield{author}{\bibinfo{person}{Qiang Huang}, \bibinfo{person}{Jianlin
  Feng}, \bibinfo{person}{Yikai Zhang}, \bibinfo{person}{Qiong Fang}, {and}
  \bibinfo{person}{Wilfred Ng}.} \bibinfo{year}{2015}\natexlab{}.
\newblock \showarticletitle{Query-aware locality-sensitive hashing for
  approximate nearest neighbor search}.
\newblock \bibinfo{journal}{\emph{PVLDB}} \bibinfo{volume}{9},
  \bibinfo{number}{1} (\bibinfo{year}{2015}), \bibinfo{pages}{1--12}.
\newblock


\bibitem[\protect\citeauthoryear{Indyk and Motwani}{Indyk and Motwani}{1998}]%
        {indyk1998approximate}
\bibfield{author}{\bibinfo{person}{Piotr Indyk} {and} \bibinfo{person}{Rajeev
  Motwani}.} \bibinfo{year}{1998}\natexlab{}.
\newblock \showarticletitle{Approximate nearest neighbors: towards removing the
  curse of dimensionality}. In \bibinfo{booktitle}{\emph{STOC}}.
  \bibinfo{pages}{604--613}.
\newblock


\bibitem[\protect\citeauthoryear{Jagadish, Ooi, Tan, Yu, and Zhang}{Jagadish
  et~al\mbox{.}}{2005}]%
        {jagadish2005idistance}
\bibfield{author}{\bibinfo{person}{Hosagrahar~V Jagadish},
  \bibinfo{person}{Beng~Chin Ooi}, \bibinfo{person}{Kian-Lee Tan},
  \bibinfo{person}{Cui Yu}, {and} \bibinfo{person}{Rui Zhang}.}
  \bibinfo{year}{2005}\natexlab{}.
\newblock \showarticletitle{iDistance: An adaptive B+-tree based indexing
  method for nearest neighbor search}.
\newblock \bibinfo{journal}{\emph{TODS}} \bibinfo{volume}{30},
  \bibinfo{number}{2} (\bibinfo{year}{2005}), \bibinfo{pages}{364--397}.
\newblock


\bibitem[\protect\citeauthoryear{Jegou, Douze, and Schmid}{Jegou
  et~al\mbox{.}}{2010}]%
        {jegou2010product}
\bibfield{author}{\bibinfo{person}{Herve Jegou}, \bibinfo{person}{Matthijs
  Douze}, {and} \bibinfo{person}{Cordelia Schmid}.}
  \bibinfo{year}{2010}\natexlab{}.
\newblock \showarticletitle{Product quantization for nearest neighbor search}.
\newblock \bibinfo{journal}{\emph{TPAMI}} \bibinfo{volume}{33},
  \bibinfo{number}{1} (\bibinfo{year}{2010}), \bibinfo{pages}{117--128}.
\newblock


\bibitem[\protect\citeauthoryear{Katayama and Satoh}{Katayama and
  Satoh}{1997}]%
        {katayama1997sr}
\bibfield{author}{\bibinfo{person}{Norio Katayama} {and}
  \bibinfo{person}{Shin'ichi Satoh}.} \bibinfo{year}{1997}\natexlab{}.
\newblock \showarticletitle{The SR-tree: An index structure for
  high-dimensional nearest neighbor queries}.
\newblock \bibinfo{journal}{\emph{ACM SIGMOD Record}} \bibinfo{volume}{26},
  \bibinfo{number}{2} (\bibinfo{year}{1997}), \bibinfo{pages}{369--380}.
\newblock


\bibitem[\protect\citeauthoryear{Kleinberg}{Kleinberg}{1997}]%
        {kleinberg1997two}
\bibfield{author}{\bibinfo{person}{Jon~M Kleinberg}.}
  \bibinfo{year}{1997}\natexlab{}.
\newblock \showarticletitle{Two algorithms for nearest-neighbor search in high
  dimensions}. In \bibinfo{booktitle}{\emph{STOC}}, Vol.~\bibinfo{volume}{97}.
  \bibinfo{pages}{599--608}.
\newblock


\bibitem[\protect\citeauthoryear{Lei, Huang, Kankanhalli, and Tung}{Lei
  et~al\mbox{.}}{2019}]%
        {lei2019sublinear}
\bibfield{author}{\bibinfo{person}{Yifan Lei}, \bibinfo{person}{Qiang Huang},
  \bibinfo{person}{Mohan Kankanhalli}, {and} \bibinfo{person}{Anthony Tung}.}
  \bibinfo{year}{2019}\natexlab{}.
\newblock \showarticletitle{Sublinear Time Nearest Neighbor Search over
  Generalized Weighted Space}. In \bibinfo{booktitle}{\emph{ICML}}.
  \bibinfo{pages}{3773--3781}.
\newblock


\bibitem[\protect\citeauthoryear{Liu, Cui, Huang, Li, and Shen}{Liu
  et~al\mbox{.}}{2014}]%
        {liu2014sk}
\bibfield{author}{\bibinfo{person}{Yingfan Liu}, \bibinfo{person}{Jiangtao
  Cui}, \bibinfo{person}{Zi Huang}, \bibinfo{person}{Hui Li}, {and}
  \bibinfo{person}{Heng~Tao Shen}.} \bibinfo{year}{2014}\natexlab{}.
\newblock \showarticletitle{SK-LSH: an efficient index structure for
  approximate nearest neighbor search}.
\newblock \bibinfo{journal}{\emph{PVLDB}} \bibinfo{volume}{7},
  \bibinfo{number}{9} (\bibinfo{year}{2014}), \bibinfo{pages}{745--756}.
\newblock


\bibitem[\protect\citeauthoryear{Lv, Josephson, Wang, Charikar, and Li}{Lv
  et~al\mbox{.}}{2007}]%
        {lv2007multi}
\bibfield{author}{\bibinfo{person}{Qin Lv}, \bibinfo{person}{William
  Josephson}, \bibinfo{person}{Zhe Wang}, \bibinfo{person}{Moses Charikar},
  {and} \bibinfo{person}{Kai Li}.} \bibinfo{year}{2007}\natexlab{}.
\newblock \showarticletitle{Multi-probe LSH: efficient indexing for
  high-dimensional similarity search}. In \bibinfo{booktitle}{\emph{VLDB}}.
  \bibinfo{pages}{950--961}.
\newblock


\bibitem[\protect\citeauthoryear{Malkov and Yashunin}{Malkov and
  Yashunin}{2018}]%
        {malkov2018efficient}
\bibfield{author}{\bibinfo{person}{Yury~A Malkov} {and}
  \bibinfo{person}{Dmitry~A Yashunin}.} \bibinfo{year}{2018}\natexlab{}.
\newblock \showarticletitle{Efficient and robust approximate nearest neighbor
  search using hierarchical navigable small world graphs}.
\newblock \bibinfo{journal}{\emph{TPAMI}} (\bibinfo{year}{2018}).
\newblock


\bibitem[\protect\citeauthoryear{Manber and Myers}{Manber and Myers}{1993}]%
        {manber1993suffix}
\bibfield{author}{\bibinfo{person}{Udi Manber} {and} \bibinfo{person}{Gene
  Myers}.} \bibinfo{year}{1993}\natexlab{}.
\newblock \showarticletitle{Suffix arrays: a new method for on-line string
  searches}.
\newblock \bibinfo{journal}{\emph{SICOMP}} \bibinfo{volume}{22},
  \bibinfo{number}{5} (\bibinfo{year}{1993}), \bibinfo{pages}{935--948}.
\newblock


\bibitem[\protect\citeauthoryear{Panigrahy}{Panigrahy}{2006}]%
        {panigrahy2006entropy}
\bibfield{author}{\bibinfo{person}{Rina Panigrahy}.}
  \bibinfo{year}{2006}\natexlab{}.
\newblock \showarticletitle{Entropy based nearest neighbor search in high
  dimensions}. In \bibinfo{booktitle}{\emph{SODA}}.
  \bibinfo{pages}{1186--1195}.
\newblock


\bibitem[\protect\citeauthoryear{Sun, Wang, Qin, Zhang, and Lin}{Sun
  et~al\mbox{.}}{2014}]%
        {sun2014srs}
\bibfield{author}{\bibinfo{person}{Yifang Sun}, \bibinfo{person}{Wei Wang},
  \bibinfo{person}{Jianbin Qin}, \bibinfo{person}{Ying Zhang}, {and}
  \bibinfo{person}{Xuemin Lin}.} \bibinfo{year}{2014}\natexlab{}.
\newblock \showarticletitle{SRS: solving c-approximate nearest neighbor queries
  in high dimensional euclidean space with a tiny index}.
\newblock \bibinfo{journal}{\emph{PVLDB}} \bibinfo{volume}{8},
  \bibinfo{number}{1} (\bibinfo{year}{2014}), \bibinfo{pages}{1--12}.
\newblock


\bibitem[\protect\citeauthoryear{Tao, Yi, Sheng, and Kalnis}{Tao
  et~al\mbox{.}}{2009}]%
        {tao2009quality}
\bibfield{author}{\bibinfo{person}{Yufei Tao}, \bibinfo{person}{Ke Yi},
  \bibinfo{person}{Cheng Sheng}, {and} \bibinfo{person}{Panos Kalnis}.}
  \bibinfo{year}{2009}\natexlab{}.
\newblock \showarticletitle{Quality and efficiency in high dimensional nearest
  neighbor search}. In \bibinfo{booktitle}{\emph{SIGMOD}}.
  \bibinfo{pages}{563--576}.
\newblock


\bibitem[\protect\citeauthoryear{Terasawa and Tanaka}{Terasawa and
  Tanaka}{2007}]%
        {terasawa2007spherical}
\bibfield{author}{\bibinfo{person}{Kengo Terasawa} {and}
  \bibinfo{person}{Yuzuru Tanaka}.} \bibinfo{year}{2007}\natexlab{}.
\newblock \showarticletitle{Spherical lsh for approximate nearest neighbor
  search on unit hypersphere}. In \bibinfo{booktitle}{\emph{Workshop on
  Algorithms and Data Structures}}. \bibinfo{pages}{27--38}.
\newblock


\bibitem[\protect\citeauthoryear{Wang, Shrivastava, Wang, and Ryu}{Wang
  et~al\mbox{.}}{2018}]%
        {wang2018randomized}
\bibfield{author}{\bibinfo{person}{Yiqiu Wang}, \bibinfo{person}{Anshumali
  Shrivastava}, \bibinfo{person}{Jonathan Wang}, {and} \bibinfo{person}{Junghee
  Ryu}.} \bibinfo{year}{2018}\natexlab{}.
\newblock \showarticletitle{Randomized Algorithms Accelerated over CPU-GPU for
  Ultra-High Dimensional Similarity Search}. In
  \bibinfo{booktitle}{\emph{SIGMOD}}. \bibinfo{pages}{889--903}.
\newblock


\bibitem[\protect\citeauthoryear{Weber, Schek, and Blott}{Weber
  et~al\mbox{.}}{1998}]%
        {weber1998quantitative}
\bibfield{author}{\bibinfo{person}{Roger Weber}, \bibinfo{person}{Hans-J{\"o}rg
  Schek}, {and} \bibinfo{person}{Stephen Blott}.}
  \bibinfo{year}{1998}\natexlab{}.
\newblock \showarticletitle{A quantitative analysis and performance study for
  similarity-search methods in high-dimensional spaces}. In
  \bibinfo{booktitle}{\emph{VLDB}}, Vol.~\bibinfo{volume}{98}.
  \bibinfo{pages}{194--205}.
\newblock


\bibitem[\protect\citeauthoryear{Zheng, Guo, Tung, and Wu}{Zheng
  et~al\mbox{.}}{2016}]%
        {zheng2016lazylsh}
\bibfield{author}{\bibinfo{person}{Yuxin Zheng}, \bibinfo{person}{Qi Guo},
  \bibinfo{person}{Anthony~KH Tung}, {and} \bibinfo{person}{Sai Wu}.}
  \bibinfo{year}{2016}\natexlab{}.
\newblock \showarticletitle{Lazylsh: Approximate nearest neighbor search for
  multiple distance functions with a single index}. In
  \bibinfo{booktitle}{\emph{SIGMOD}}. \bibinfo{pages}{2023--2037}.
\newblock


\bibitem[\protect\citeauthoryear{Zhou, Guo, Jagadish, Krcal, Liu, Luan, Tung,
  Yang, and Zheng}{Zhou et~al\mbox{.}}{2018}]%
        {zhou2018generic}
\bibfield{author}{\bibinfo{person}{Jingbo Zhou}, \bibinfo{person}{Qi Guo},
  \bibinfo{person}{HV Jagadish}, \bibinfo{person}{Lubos Krcal},
  \bibinfo{person}{Siyuan Liu}, \bibinfo{person}{Wenhao Luan},
  \bibinfo{person}{Anthony~KH Tung}, \bibinfo{person}{Yueji Yang}, {and}
  \bibinfo{person}{Yuxin Zheng}.} \bibinfo{year}{2018}\natexlab{}.
\newblock \showarticletitle{A generic inverted index framework for similarity
  search on the GPU}. In \bibinfo{booktitle}{\emph{ICDE}}.
  \bibinfo{pages}{893--904}.
\newblock


\end{thebibliography}

\end{document}